\documentclass[a4paper]{article}
    \usepackage[a4paper,top=3cm,bottom=3cm,left=2.1cm,right=2.1cm,marginparwidth=1.75cm]{geometry}
    \usepackage[utf8]{inputenc}
    \usepackage{bm}
    \usepackage{amsmath, amsfonts, amssymb, amsthm}
    \usepackage[justification=centering]{caption}
    \usepackage{mathrsfs, dsfont}
    \usepackage[dvipsnames]{xcolor}
    \usepackage{graphicx}
    \usepackage{float}
    \usepackage{subfig}
    \usepackage{bigints}

    \usepackage[colorlinks=true,allcolors=blue]{hyperref}
    \usepackage{xcolor}
    \usepackage{xparse} 
    \usepackage{hyperref}
    \usepackage{multirow}
    \usepackage{parskip}

    \usepackage{tasks}
    \settasks{style=itemize}
    
    \usepackage[round]{natbib}
    \usepackage{pifont}
    \newtheorem{theorem}{Theorem}[section]
    \newtheorem{lemma}[theorem]{Lemma}
    \newtheorem{definition}[theorem]{Definition}
    
    \newtheorem{proposition}[theorem]{Proposition}
    
    \newtheorem{corollary}[theorem]{Corollary}

    \theoremstyle{remark}
    \newtheorem{remark}{Remark}[theorem]
    \newtheorem{example}{Example}[theorem]

    \numberwithin{equation}{section}
    
    \newenvironment{sqremark}{\begin{remark}}{\hfill \tiny $\blacksquare$ \end{remark}}

    \newcommand{\R}{\mathbb{R}}
    \newcommand{\N}{\mathbb{N}}
    
    \newcommand{\E}{\mathbb{E}}
    \def\P{\mathbb{P}}
    
    \newcommand\F{\mathcal{F}}

    \newcommand{\indic}[1]{\mathds{1}_{\left\{ #1 \right\}}}

    \newcommand{\alphabet}[1][d]{{A}_{#1}}
    \newcommand{\TA}[1][d]{T(\R^{#1})}
    \newcommand{\eTA}[1][d]{T((\R^{#1}))}
    \newcommand{\tTA}[2][d]{T^{#2}(\R^{#1})}
    
    \newcommand{\word}[1]{{\mathcolor{NavyBlue}{\mathbf{#1}}}}
    \newcommand{\emptyword}{{\color{NavyBlue}\textup{\textbf{\o{}}}}}
    \newcommand{\proj}[1]{|_{\word{#1}}}
    
    \newcommand{\conpow}[1]{^{\otimes #1}}

    \newcommand{\shuprod}{\mathrel{\sqcup \mkern -3.2mu \sqcup}}
    \newcommand{\shupow}[1]{^{\shuprod #1}}

    \NewDocumentCommand{\sig}{O{t} O{W}}{\widehat{\mathbb{#2}}_{#1}}
    \NewDocumentCommand{\sigE}{O{t} O{W}}{\E[\sig[#1][#2]]}
    \NewDocumentCommand{\sigX}{O{t} O{X}}{\mathbb{#2}_{#1}}

    \NewDocumentCommand{\bracketsigX}{O{t} O{X} m}{\left \langle #3, \sigX[#1][#2] \right \rangle} 
    \NewDocumentCommand{\bracketsig}{O{t} O{W} m}{\left \langle #3, \sig[#1][#2] \right \rangle}   
    \NewDocumentCommand{\bracketsigE}{O{t} O{W} m}{\left \langle #3, \sigE[#1][#2] \right \rangle} 
    
    \newcommand{\bsigma}{\bm{\sigma}}

    \newcommand{\bell}{\bm{\ell}}
    \newcommand{\bp}{\bm{p}}

    \newcommand{\sign}[1]{\textnormal{sign} \left( #1 \right)}
    
    \usepackage{mathtools}
    \usepackage{authblk}
    \mathtoolsset{showonlyrefs}

    \usepackage{enumitem}
    \usepackage[normalem]{ulem}
    \usepackage[textsize=small]{todonotes}
    
    \title{Martingale property and moment explosions in signature volatility models}
    
    \author[1]{Eduardo Abi Jaber\thanks{eduardo.abi-jaber@polytechnique.edu. The first author is grateful for the financial support from the Chaires FiME-FDD, Financial Risks, Deep Finance \& Statistics and Machine Learning and systematic methods in finance at École Polytechnique.}}
    \author[2]{Paul Gassiat\thanks{gassiat@ceremade.dauphine.fr}}
    \author[1,3]{Dimitri Sotnikov\thanks{dmitrii.sotnikov@polytechnique.edu. The author is grateful for the financial support provided by Engie Global Markets.}}
    \affil[1]{École Polytechnique, CMAP}
    \affil[2]{CEREMADE, Université Paris Dauphine-PSL $\&$ DMA, ENS Paris $\&$ Institut Universitaire de France}
    \affil[3]{Engie Global Markets}

\begin{document}

\maketitle

\begin{abstract}
We study the martingale property and moment explosions of a signature volatility model, where the volatility process of the log-price is given by a linear form of the signature of a time-extended Brownian motion. Excluding trivial cases, we demonstrate that the price process is a true martingale if and only if the order of the linear form is odd and a correlation parameter  is negative.  The proof involves a fine analysis of the explosion time of a signature stochastic differential equation. This result is of key practical relevance, as it highlights that, when used for approximation purposes, the linear combination of signature elements must be taken of odd order to preserve the martingale property. Once martingality is established, we also characterize the existence of higher moments of the price process in terms of  a condition on a correlation parameter.
\end{abstract}

\section{Introduction}
The martingale property and moment explosions of the asset price process are fundamental considerations  in stochastic volatility modeling, particularly due to their implications for derivative pricing and numerical methods. The martingale property is a key criterion for assessing the validity of a stochastic volatility model, as emphasized by \citet{sepp2023robust_SV}. In particular, strict local martingale models exhibit structural deficiencies, such as the failure of put-call parity and deviations in the asymptotic behavior of implied volatility wings \citep*{cox2005local, heston2007options, jacquier2015impliedvolatilitystrictlocal}.  Similarly, the existence of moments for the price process plays a crucial role in both theoretical and computational aspects of option pricing. In particular, the existence of finite moments is crucial for ensuring the numerical stability and convergence of Monte Carlo methods used to compute derivative prices and their sensitivities. Moments also play a key role in the proof of asymptotic formulae, allowing one to transition from large deviations of the price process to call price asymptotics, see for instance \cite[Section 4.2]{FGP21} { and \cite{lee2004moment}}.
 
Both of these problems have been extensively studied in the literature, dating back to the works of \citet{Girsanov1960} and \citet{Novikov1973}. However, Novikov's criterion is often not applicable to stochastic volatility models. { An alternative approach consists in using a change-of-measure argument to transform the problem into the study of SDE explosion; see, for instance, \citet[Chapter 3.7]{McKean1969}. This technique was applied to stochastic volatility modeling by \citet{Sin98}, who observed that in certain models}---such as the Hull--White model---the stock price can be a strict martingale for certain parameter values. In this context, \citet{Lions2007} present general sufficient conditions for martingality and strict local martingality when the volatility process follows a stochastic differential equation (SDE), and also address the problem of moment existence. Subsequently, \citet{Mijatovic2012} provide a more general criterion for martingality in SDE-based models, while \citet*{bernard2017martingale} extend this work to the general spot-volatility correlation case $-1 \leq \rho \leq 1$, offering a complete classification of convergence properties for general one-factor Markovian stochastic volatility models.

Most of these works, however, focus on Markovian models where both the asset price and volatility processes are solutions of SDEs. For example, \citet{Jourdain2004} proposes a martingality criterion, along with integrability conditions for prices in SABR-like models. Similarly, \citet{Andersen2007} address the moment problem and provide sufficient martingality conditions for SABR-like and Heston-like models.

{ For the non-Markovian setting, we refer to the works of \citet{Criens2020NoAI} and \citet{Criens2018}, which establish general results on the martingale property of stochastic exponential martingales for Itô processes. However, deriving such results for specific path-dependent models often requires careful, case-by-case analysis. A notable class of such models considers volatility processes governed by stochastic Volterra equations.} For the rough Bergomi model, \citet{Gassiat2018OnTM} establishes conditions for both martingality and non-martingality, as well as a sufficient condition for moment explosions. \citet*{gerhold2024} propose an alternative approach to verifying higher-order moment explosions by examining the integrability of the price supremum and apply this method to the rough Bergomi model. 
In the context of the Volterra Heston model, \citet*{jaber2019affinevolterraprocesses} establish conditions for the martingale property, while \citet*{Gerhold2018} study the problem of moment explosions. In the class of so-called path-dependent volatility (PDV) models \citep{Guyon02092023},  \citet{Nutz2024} analyze the martingale property in a two-factor PDV model. Additionally, \citet{Ruf_2015} provides a general martingality criterion for stochastic exponentials defined as non-anticipative path-dependent functionals of solutions to stochastic differential equations. {Despite its generality, this result cannot be directly applied to stochastic volatility models and requires a thorough analysis to verify the condition for each particular model.}

Recently, signature-based models have gained popularity in non-Markovian modeling. Such models were first studied by \citet*{arribas2020sig} and further developed by \citet*{cuchiero2022theocalib, cuchiero2023spvix}, and \citet{abijaber2024signature}. Their distinguishing feature is that the volatility process is represented as a (possibly infinite) linear combination of signature elements. These models can serve both for the approximation, see \citet{cuchiero2022theocalib}, and the exaxt representation, see \citet{abijaber2024signature}, of a wide range of  models, though in both cases, the linear combination of signature elements must be truncated for practical use. As we will demonstrate, the truncation order plays a crucial role in determining whether the price process is a true martingale. However, no results on either martingality or moment explosions have been established for these models so far. The aim of this paper is to fill this gap by providing necessary and sufficient conditions for the martingality of the price process and by studying the finiteness of its moments in signature volatility models. We adopt the multiplicative framework of \citet{abijaber2024signature, cuchiero2023spvix}, in which the price process is given by a stochastic exponential, making the question of martingality non-trivial.

The Markovian setting provides key insights into the martingality problem. In particular, for truncated exponential Ornstein-Uhlenbeck models with negative spot-volatility correlation, obtained by truncating the power series expansion of the exponential, the price process is a martingale if and only if   the truncation order is odd. A similar pattern appears in polynomial volatility models, where volatility is expressed as a polynomial function of a Brownian motion or an Ornstein–Uhlenbeck process, as in the Quintic model of \citet*{quintic2022}. {These results can be established using techniques similar to those in \citet[Theorem 2.4(i)]{Lions2007}.} Since path signatures play an analogous role to polynomials on path space, we expect similar characterizations to hold beyond the Markovian framework. 

In this work, we show that the order of the linear combination of signature elements plays the same role as the order of a polynomial and that a similar result holds for signature volatility models. Specifically, given two Brownian motions, $B$ and $W$, correlated with coefficient $\rho \in [-1, 1]$, we consider the following signature volatility model for a price process $S$:  
\begin{align}  
    \dfrac{dS_t}{S_t} &= \sigma_t dB_t, \quad t \geq 0, \\  
    \sigma_t &= \bracketsig{\bsigma}, 
\end{align}  
where the volatility $\sigma_t$ is a finite linear combination of order $N$ in the signature $ \sig $ of the extended Brownian motion $ \widehat{W} = (t, W) $, defined as a sequence of iterated Stratonovich integrals, with the first orders given by
\begin{equation}
    \sig^0 = 1,
    \quad
    \sig^1 =
    \begin{pmatrix}
        t \\
         W_t
    \end{pmatrix},
    \quad
    \sig^2 =
    \begin{pmatrix}
        \frac{t^2}{2!} & \int_0^t s d W_s \\
        \int_0^t W_s ds & \frac{W_t^2}{2!}
    \end{pmatrix}, \quad \ldots
\end{equation}

\textbf{The martingale property.} Our main result,  Theorem~\ref{T:main_martingality}, establishes a very simple necessary and sufficient condition for the martingality of the price process $S$ under the assumption that the leading coefficient $\sigma^{\word{2}\conpow{N}}$ in front of the signature element $\frac{W_t^N}{N!} $ is nonzero. More precisely, in the nontrivial case where $ N \geq 2$ and $\rho \neq 0 $, we have that 
\begin{center}  
\textit{the price process \( S \) is a true martingale if and only if \( N \) is odd and \( \rho\sigma^{\word{2}\conpow{N}} \leq 0 \).}  
\end{center}  

The standard technique for verifying the martingale property reduces to studying the explosion time of a certain  stochastic differential equation (SDE) following a change of probability measure. In our case, this yields a signature SDE whose explosion time is intricate to analyze, see Proposition~\ref{prop:sigODE} and Theorem~\ref{T:explosion_criterion}. To address this, we establish two key bounds for linear combinations of signature elements in Lemmas~\ref{lemma:linear_from_upper_bound} and \ref{lem:expYV}. These bounds exploit the fact that the shuffle algebra forms a polynomial algebra over the \citet{Lyndon1954} words, thanks to the \citet{Radford1979ANR} Theorem. Using these bounds, we then prove the finiteness or explosion of the solution of the signature SDE, depending on the parity of the order $N$ and the sign of $\rho$, leading to Theorem~\ref{T:main_martingality}. We stress that such result is of key practical relevance, as it highlights that, when used for approximation purposes, the linear combination of signature elements must be taken of odd order to preserve the martingale property. The non-positivity of $\rho \sigma^{\word{2}\conpow{N}}$ does not imply that the spot-volatility correlation is necessarily non-positive, see Remark~\ref{R:spotvolcorrel}. 

Additionally, we extend the martingale characterization to a volatility process constructed from the signature of a multidimensional Brownian motion, as presented in Theorem~\ref{T:main_martingality_multid}.

\textbf{Moment explosions.} Assuming the conditions for $S$ to be a martingale, we prove in Theorem~\ref{thm:moments}  that
\begin{center}
    \textit{the moment $\E[S_T^m]$ of order $m$ is finite for any $T>0$ if $|\rho| > \sqrt{1 - \frac{1}{m}}$. \\ Conversely, if $|\rho| <  \sqrt{1 - \frac{1}{m}}$, then $\E[S_T^m] = +\infty$ for all $T > 0$.}
\end{center}

This result aligns well with those obtained in the Markovian case, such as \cite[Theorem 2.3]{Lions2007} and \cite[Proposition 6]{Jourdain2004}. The same condition for moment explosions in the non-Markovian rough Bergomi model was proved in \cite{Gassiat2018OnTM}, {see also \cite{Gul20}}. {To prove Theorem~\ref{thm:moments}, {we rely on the \citet{BD98} formula}, which reduces the moment existence problem to the finiteness of an optimal control problem value. In the case of infinite moments, we explicitly construct a sequence of controls to show that the problem value is infinite, whereas in the case of finite moments, we establish the finiteness of the value function using the bounds on linear combinations derived earlier. { For the critical case   $|\rho| = \sqrt{1 - \frac{1}{m}}$,  we show in Proposition~\ref{prop:crit} that the finiteness of the $m$-th moment may or may not hold depending on the other coefficients of $\bsigma$ and the time horizon $T$.}
}

\paragraph{Outline.}
Section~\ref{sec:signature} introduces the necessary definitions and key concepts related to tensor algebra and signatures. In Section~\ref{section:main_res}, we present our main results concerning the martingality criterion and moment explosions for the signature volatility model. { Section~\ref{section:general_crit} demonstrates how the martingality problem can be reduced to the explosion problem of a signature SDE. Section~\ref{sect:bounds_proofs} establishes two essential bounds on signature words, which are extensively used in the subsequent proofs.} Sections~\ref{sect:martingality_expec} and~\ref{section:proof_local_martingality} are dedicated to proving the martingality and strict local martingality criteria, respectively. Finally, Section~\ref{sect:moments} contains the proof concerning moment explosions.

\section{Preliminary concepts of path signatures} \label{sec:signature}

\subsection{Tensor algebra}

    In this section, we introduce the notation and definitions necessary for dealing with signatures of continuous semimartingales. We also refer the reader to the introductory sections in \cite{cuchiero2023spvix, abijaber2024signature}.
    
    Fix $d \in \N$. For $n\in\N$, let $(\R^d) \conpow{n}$ denote the space of tensors of order $n$ and $(\R^d) \conpow{0} = \R$. 
    The path signatures defined below, will be defined as tensor sequences, i.e.~the elements of the extended tensor algebra space $\eTA $ over $\R^d$ defined by
    $$ \eTA := \left\{ \bell = (\bell^n)_{n=0}^\infty : \bell^n \in (\R^d) \conpow{n} \right\}. $$
    The truncated tensor algebra $\tTA{N}$ is defined as the space of sequences such that all their elements of order greater than $N\in\N$ are equal to zero. We will also denote by $\TA$ the space of all finite tensor sequences, i.e.~$\TA := \bigcup_{N \in \N } \tTA{N}$,
    and we note than $\TA \subset \eTA$.

    Let $\{ e_1, \dots, e_d \} \subset \R^d$ denote the canonical basis of $\mathbb{R}^d$. Then, $(e_{i_1} \otimes \cdots \otimes e_{i_n})_{(i_1, \dots, i_n) \in \{ 1, \dots, d \}^n}$ is a basis of $(\R^d) \conpow{n}$. To simplify the notations, we identify its elements $e_{i_1} \otimes \cdots \otimes e_{i_n}$ with the words $\word{i_1 \cdots i_n}$ of length $n$ in the alphabet $\alphabet = \{ \word{1}, \word{2}, \dots, \word{d} \}$. All the words of length $n$ will be denoted by
    \begin{equation} \label{eq:sig_basis}
        V_n := \{ \word{i_1 \cdots i_n}: \word{i_k} \in \alphabet \text{ for } k = 1, 2, \dots, n \}.
    \end{equation}
    
    We adopt the convention of denoting by $\emptyword$ the empty word and by $V_0 = \{ \emptyword \}$ identified with the basis of $(\R^d) \conpow{0} = \R$. The basis of $\eTA$ can be then identified with $V := \cup_{n \geq 0} V_n$, and each sequence $\bell \in \eTA$ can be written as
    \begin{equation} \label{eq:sig_expansion}
        \bell = \sum_{n=0}^\infty \sum_{\word{v} \in V_n} \bell^{\word{v}} \word{v},
    \end{equation}
    where $\bell^\word{v} = (\bell^{n})^\word{v} \in\R$ is the coefficient of $\bell$ corresponding to the word $\word{v}$ and $n=|\word{v}|$ stands for its length. The concatenation $\bell \word{v}$ of elements $\bell \in \eTA$ and the word $\word{v} = \word{i_1 \cdots i_n}$ means $\bell \otimes e_{i_1} \otimes \cdots \otimes e_{i_n}$.
    
    We also define the bracket between $\bell \in \TA$ and $\bp \in \eTA$ by
    \begin{align} \label{eq:bracket}
        \langle \bell, \bp \rangle 
        := \sum_{n=0}^{\infty} \sum_{\word{v} \in V_n} \bell^\word{v} \bp^\word{v}. 
    \end{align}
    Note that since $\bell$ is finite, the sum contains only a finite number of terms, so that the bracket is always well-defined.

    Another key operation on the space of tensor sequences is the \textit{shuffle product}, which generalizes the integration by parts formula and plays a crucial role in the signature linearization property stated in Proposition~\ref{prop:shufflepropertyextended} below.
    \begin{definition}[Shuffle product] \label{def:shuffleprod}
        For two words $\word{v}, \word{w} \in V$ and two letters $\word{i},\word{j}\in\alphabet$, the shuffle product $\shuprod: V \times V \to \TA$ is defined recursively by
        \begin{align*}
            (\word{v} \word{i}) \shuprod (\word{w} \word{j}) &
            = (\word{v} \shuprod (\word{w} \word{j})) \word{i} + ((\word{v} \word{i}) \shuprod \word{w}) \word{j},
            \\ \word{w} \shuprod \emptyword &
            = \emptyword \shuprod \word{w} = \word{w}.
        \end{align*}
        It extends to $\eTA$ by linearity: for $\bell = \sum\limits_{n=0}^\infty \sum\limits_{\word{v} \in V_n} \bell^{\word{v}} \word{v}$ and $\bp = \sum\limits_{n=0}^\infty \sum\limits_{\word{v} \in V_n} \bp^{\word{v}} \word{v}$, the shuffle product is given by 
        \begin{align*}
            \bell\shuprod\bp = \sum_{n=0}^\infty \sum_{\substack{\word{v}, \word{w} \in V \\ |\word{v}| + |\word{w}| = n}}\bell^{\word{v}}\bp^{\word{w}}\word{v}\shuprod\word{w}.
        \end{align*}
        
        Note that the shuffle product is commutative. 
    \end{definition}
        The shuffle product represents all possible riffle shuffles of two separate decks of cards, where the relative order within each individual deck is preserved. For instance,
    $
    \word{12} \shuprod \word{3} = \word{123} + \word{132} + \word{312}.
    $
    For more details on the shuffle product, we refer to \cite{reeshuffles} and \cite{gainesshuffle}. 
    
\subsection{Signatures}
    Let $X$ and $Y$ be two { continuous} semimartingales. Throughout the paper, we denote the Itô integral by $\int_0^\cdot Y_t dX_t$ and the Stratonovich integral by $\int_0^\cdot Y_t \circ dX_t.$ The following relation holds:
    $$\int_0^\cdot Y_t \circ dX_t = \int_0^\cdot Y_t dX_t  + \frac{1}{2} [X,Y]_\cdot.$$

    \begin{definition}[Signature] \label{def:sig}
        Let $(X_t)_{t \geq 0}$ be a continuous $\R^d$-valued semimartingale on a filtered probability space $(\Omega, \F, (\F_t)_{t \geq 0}, \P)$ and fix $T > 0$. The signature of $X$ is defined by  
        \begin{align*}
            \mathbb{X}: \Omega \times [0, T] &
            \to \eTA
            \\ (\omega, t) &
            \mapsto \sigX (\omega) := (1, \sigX^1(\omega), \dots, \sigX^n(\omega), \dots),
        \end{align*}
        where
        $$ \sigX^n := \int_{0 < u_1 < \cdots < u_n < t} \circ d X_{u_1} \otimes \cdots \otimes \circ d X_{u_n} \in (\R^d)^{\otimes n}, \quad n \geq 0, \quad t\in[0, T].$$
    \end{definition}
    
    The signature can be seen as a set of monomials in the path-space, while linear combinations of signature elements $\bracketsigX{\bell}$ play the role of polynomials. For instance, if $d=1$, the signature of $X$ is exactly the sequence of monomials $\left( \frac{1}{n!} (X_t - X_0)^n \right)_{n \in \N}$ and any finite combination $\bracketsigX{\bell}$ for $\bell \in \tTA{N}$ is a polynomial of degree $N$ in $X_t$.
    
    \begin{remark}\label{rmk:sig_iteration_def}
        Each element of the signature $\sigX^n = (\sigX^\word{i_1 \cdots i_n})_{(\word{i_1 \cdots i_n}) \in V_n}$ can be written coordinate-wise in an iterative form
        \begin{align} \label{eq:signdef2}
            \sigX^\word{i_1 \cdots i_n} = \int_0^t \sigX[s]^\word{i_1 \cdots i_{n-1}} \circ d X_s^{\word{i_n}}.
        \end{align}
    \end{remark}

    \subsection{Linear combinations of signature elements}

     For the signature volatility model, we will consider the signature $\sigX[t][\widehat{X}]$ of the time-extended path $\widehat{X}_t := (t, X_t) \in \R^{d + 1}$. The main object of our study is the finite linear combination in $\sigX[t][\widehat{X}]$ given by
    \begin{equation}\label{eq:linear_comb_def}
        \bracketsigX[t][\widehat{X}]{\bell} = \sum_{n=0}^N \sum_{\word{v} \in V_n} \bell^{\word{v}} \sigX[t][\widehat{X}]^{\word{v}}, \quad \bell\in\tTA[d+1]{N},
    \end{equation}
    for some $N \in\N$.
    
    Linear combinations of signature elements have a remarkable property: they allow for the linearization of polynomials formed from these combinations. This property is frequently emphasized in the literature when highlighting the signature’s linearization power.

    \begin{proposition}[Shuffle property] \label{prop:shufflepropertyextended}
        If $\bell_1, \bell_2 \in \tTA[d + 1]{}$, then 
        $$ \bracketsigX[t][\widehat{X}]{\bell_1} \bracketsigX[t][\widehat{X}]{\bell_2} = \bracketsigX[t][\widehat{X}]{\bell_1 \shuprod \bell_2}. $$
    \end{proposition}
    
    \begin{proof}
       {See \cite*[Theorem 2.15.]{Lyons07DE}.} 
    \end{proof}

\section{Main results}\label{section:main_res}

Let $(\Omega, \mathcal{F}, (\mathcal{F}_t)_{t \geq 0}, \mathbb{P})$ be a filtered probability space supporting a two-dimensional correlated Brownian motion $(B, W)$ with correlation parameter $\rho \in [-1,\, 1]$ and with the filtration generated by this Brownian motion. We set $\bar\rho  = \sqrt{1 - \rho^2}$ and, for $\rho \in (-1, 1)$, we define
\begin{equation}
    B^\perp := \dfrac{W - \rho B}{\bar\rho}, \quad W^\perp := \dfrac{B - \rho W}{\bar\rho},
\end{equation}
the Brownian motions independent from $B$ and $W$ respectively. 

We consider a signature volatility model in the sense of \cite{abijaber2024signature}, see also \cite*{cuchiero2023spvix}. In this model, the stock price process $S$ under the risk-neutral probability $\mathbb{P}$ is given by a stochastic volatility model where the volatility process $\sigma_t$ is a finite linear combination of the elements of signature $\sig$ of the time extended Brownian motion $\widehat W_t := (t, W_t)$:
\begin{equation}\label{eq:price_model}
    \begin{aligned}
    \dfrac{dS_t}{S_t} &= \sigma_t dB_t, \quad t \geq 0, \\
    \sigma_t &= \bracketsig{\bsigma}.
    \end{aligned}
\end{equation}

Here $\bsigma \in \tTA[2]{N}$  is a tensor sequence of order $N \in \mathbb N$ in the form
\begin{equation}\label{eq:sigma_def}
    \bsigma = \sum_{k \geq 0}^{N}\sum_{|\word{v}| = k} \sigma^\word{v} \word{v}, \quad \sigma^{\word v} \in \R, 
\end{equation}

We note that 
\begin{equation}
    \E \left[\int_0^T \sigma_t^2\, dt\right] = \bracketsigE[T]{\bsigma\shupow{2}\word{1}} ,
\end{equation}
is finite, as the right-hand side involves a sum of a finite number of terms where the expected signature $\sigE[T]$ is given by Fawcett's formula \cite{fawcett} extended to the time-augmented Brownian motion in \cite[Proposition 4.10]{lyonsvictoir}.
This ensures the existence of a unique strong solution to \eqref{eq:price_model} given by the exponential local martingale
\begin{equation}
    S_t = S_0\exp\left\{-\dfrac12 \int_0^t \sigma_s^2\, ds + \int_0^t\sigma_s\,dB_s\right\}, \quad t \geq 0.
\end{equation}
We note that the volatility dynamics in the model \eqref{eq:price_model} are driven by the signature \( \sig \) of a one-dimensional Brownian motion. We will first characterize the martingale property in Subsection~\ref{sect:main_res_mart}, then study moment explosions in Subsection~\ref{sect:main_res_mom}, before extending the characterization of the martingale property to volatility dynamics driven by  a multidimensional Brownian motion in Subsection~\ref{sect:multid_extension}.

\subsection{Martingale property}\label{sect:main_res_mart}
Our main result in this section provides a necessary and sufficient condition for the price process $S$, defined by \eqref{eq:price_model}, to be a martingale. This condition is given in terms of the parity of the signature level $N$ and the correlation parameter $\rho$.  
\begin{theorem}\label{T:main_martingality}  Let $S$ be given by \eqref{eq:price_model}. 
    Suppose that the leading coefficient $\sigma^{\word{2}\conpow{N}} \not= 0$.
    \begin{enumerate}
        \item[(i)] If $\rho = 0$ or $N = 1$, then $S$   is a martingale.
        \item[(ii)] If $\rho \not= 0$ and $N \geq 2$, then $S$  is a martingale if and only if the order $N$ is odd and $\rho\sigma^{\word{2}\conpow{N}} \leq 0.$
    \end{enumerate}
    
\end{theorem}
\begin{proof}
    \textit{(i)} and sufficiency in \textit{(ii)} are proved in Section~\ref{sect:martingality_expec} and necessity in \textit{(ii)} is proved in Subsection~\ref{section:proof_local_mart_1D}.
\end{proof}

{
\begin{sqremark}
Part (i) of the theorem has already been well understood and is covered in the existing literature. Indeed, when \(\rho = 0\), the volatility process is independent of the Brownian motion \(B\) driving the price, and a conditioning argument (see, for example, \citet[Lemma 4.8]{karatzas2007numeraire}) immediately yields the martingale property. In the case \(N = 1\), the volatility process reduces to a Brownian motion with drift, and martingality follows (see, for instance, \citet[Corollary 6.1]{Criens2018}). We nevertheless include this part of the statement for completeness and to illustrate that these particular cases are also naturally covered by our approach.
\end{sqremark}
}

{The results of Theorem~\ref{T:main_martingality} can be illustrated numerically on implied volatility surfaces. Namely, if the spot price $S$ is a strict local martingale, the put-call parity is not verified, and hence, put-option implied volatility is different from the call-option implied volatility which may be not even defined, as shown by \citet{jacquier2015impliedvolatilitystrictlocal}. More precisely, recall that the Black--Scholes Call and Put option prices $C_{\mathrm{BS}}(T, K, \sigma)$ and $P_{\mathrm{BS}}(T, K, \sigma)$ satisfy the put-call parity
$$
C_{\mathrm{BS}}(T, K, \sigma) - P_{\mathrm{BS}}(T, K, \sigma) = S_0 - K,
$$
while the model Call and Put option prices $C(T, K)$ and $P(T, K)$ satisfy only
$$
C(T, K) - P(T, K) = \E[S_T] - K.
$$
Thus, if $P(T, K) = P_{\mathrm{BS}}(T, K, \sigma_{\mathrm{BS}}^P)$, one obtains
\begin{equation}
    C(T, K) = C_{\mathrm{BS}}(T, K, \sigma_{\mathrm{BS}}^P) + \E[S_T] - S_0 =  C_{\mathrm{BS}}(T, K, \sigma_{\mathrm{BS}}^C).
\end{equation}
Hence, the call and put implied volatilities coincide if and only if the put-call parity is satisfied, which holds if and only if $\mathbb{E}[S_T] = S_0$.  In turn, this is equivalent to the martingality of $S$, since $S$ is a positive local martingale and therefore a supermartingale.

We compute the put and call implied volatility smiles via the Monte Carlo method using $10^6$ trajectories of the spot price given by \eqref{eq:price_model} for even ($N = 4$) and odd ($N = 5$) orders and for positive ($\rho = 0.9$) and negative ($\rho = -0.9$) {values of the correlation parameter}. The maturity is taken equal to $T = 1$. The coefficients $\bsigma$ are chosen randomly (same random numbers were used in all cases) from the uniform distribution on $[-0.5,\, 0.5]$, except for the coefficient $\sigma^{\word{2}\conpow{N}}$ fixed equal to $1$.

Simulation results are provided in Figure \ref{fig:sig_smiles}.
As expected, the only case where call and put implied volatilities coincide corresponds to the odd order $N = 5$ and negative correlation $\rho = -0.9$.}

\begin{figure}[h]
    \begin{center}
    \includegraphics[width=0.9\linewidth]{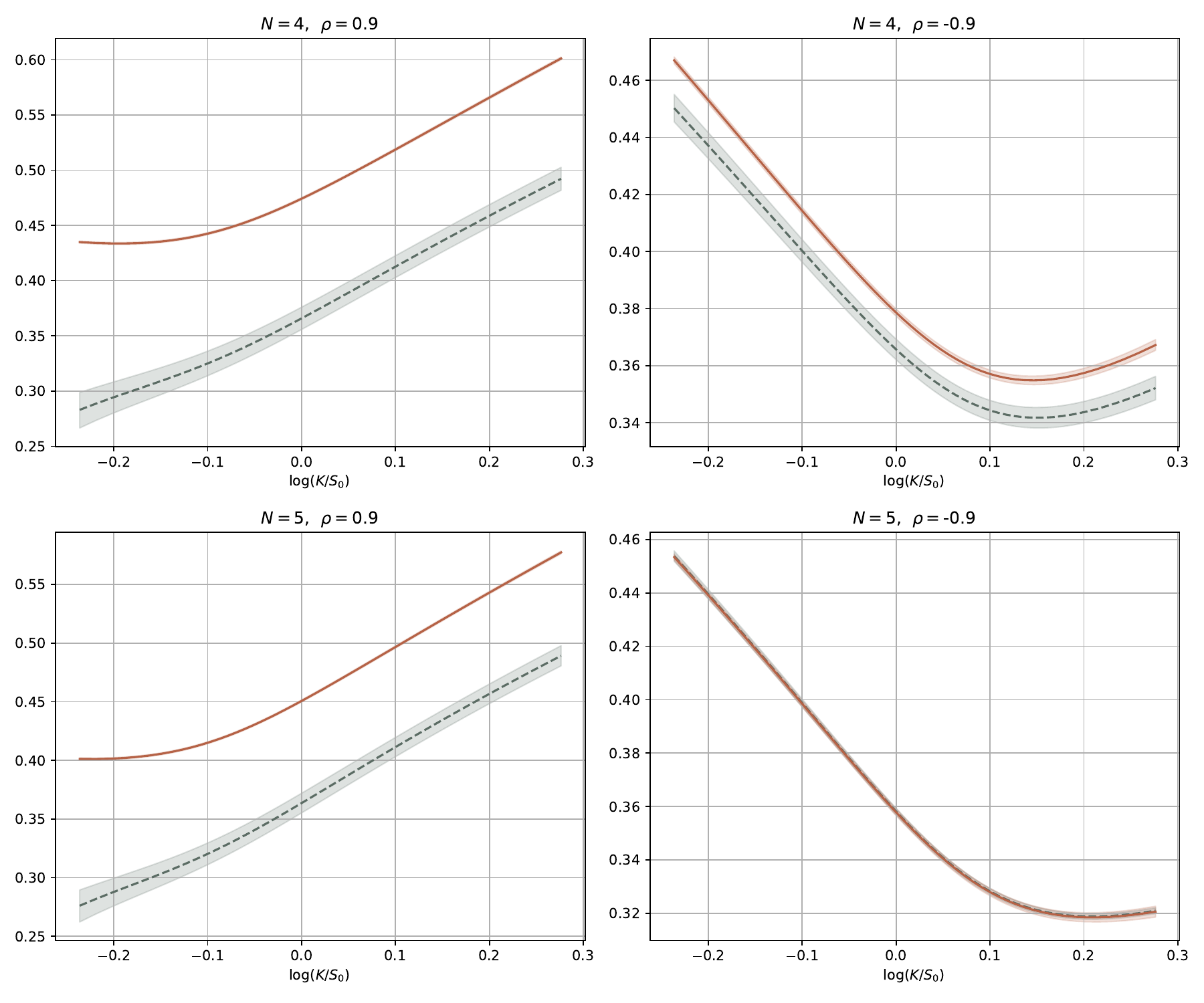}
    \caption{The put (brown) and call (dark green) implied volatility smiles for $N = 4,\ \rho = 0.9$ (upper-left), \\ $N = 4,\ \rho = -0.9$ (upper-right), $N = 5, \ \rho = 0.9$ (lower-left), and $N = 5,\  \rho = -0.9$ (lower-right). \\ The confidence intervals with confidence level of $95\%$ do not intersect anywhere except for the forth plot $N = 5,\  \rho = -0.9$, corresponding to the martingale spot price.}
    \label{fig:sig_smiles}
    \end{center}
\end{figure}

\begin{remark}
    The characterization of Theorem~\ref{T:main_martingality} still holds if, instead of $\sig$, we consider the signature of a time-extended Ornstein--Uhlenbeck process $\widehat{Y}_t = (t, Y_t)$ for the volatility process, where $Y$ satisfies  
    \[
    dY_t = a Y_t dt + b dW_t, \quad t \geq 0,
    \]
    for some $a, b \in \mathbb{R}$. Although slight modifications in the proof of Lemma~\ref{lem:expYV} are needed, the proof of this extension follows the same ideas presented later in this paper. The case of the Ornstein--Uhlenbeck process is important for practical applications as it includes the Quintic model of \citet*{quintic2022} and can be seen as a particular case of the multi-factor signature volatility model proposed by \citet*{cuchiero2023spvix}.
\end{remark}

\begin{remark}
Theorem~\ref{T:main_martingality}  provides useful insights into the practical use of the signature volatility model \eqref{eq:price_model}. Specifically, when performing a linear regression of the volatility process against the driving Brownian motion, using its exact signature representation, or calibrating the model, a suitable truncation order must be chosen. Our result offers clear and simple guidelines for selecting hyperparameters: $N$ should be chosen to be odd, while the condition on $ \rho\sigma^{\word{2}\conpow{N}} $ can be easily incorporated as a constraint in the optimization problem. Failing to do so could have dramatic consequences for martingality, as illustrated in Figure~\ref{fig:sig_smiles}.
\end{remark}

\begin{remark}\label{R:spotvolcorrel}
    We note that the condition $ \rho\sigma^{\word{2}\conpow{N}} \leq 0 $ does not imply that the spot-volatility correlation must be negative. Indeed, as shown in \cite[Subsection 3.2]{abijaber2024signature}, the spot-volatility correlation is given by
    $$
    \frac{d[\log S, |\sigma|]_t}{\sqrt{d[\log S]_t}\sqrt{d[|\sigma|]_t}} = \rho \cdot \sign{\bracketsig{\bsigma\proj{2}}},
    $$
    where $ \bsigma\proj{2} = \sum\limits_{n=0}^N \sum\limits_{|\word{v}| \leq n} \bsigma^{\word{v2}} \word{v} $. Since the sign depends on the other coefficients of $ \bsigma $, the spot-volatility correlation can be positive while still ensuring the price martingality. In this way, the model can be applied to markets with a positive implied volatility skew, such as FX or commodity markets.
\end{remark}

\subsection{Moments}\label{sect:main_res_mom}

In this subsection, we present our main result concerning the moment explosions of the process $S$ in the signature volatility model \eqref{eq:price_model}. Under the assumptions ensuring true martingality stated in Theorem~\ref{T:main_martingality}, we identify the threshold value of $m$ for which $\mathbb{E}[S_T^m]$ is finite.

\begin{theorem} \label{thm:moments} Let $S$ be given by \eqref{eq:price_model}. 
    Suppose that the order $N \geq 2$ is odd and $\rho\sigma^{\word{2}\conpow{N}} < 0$. 
    \begin{enumerate}
        \item[(i)] If $|\rho| < \sqrt{1 - \frac{1}{m}}$, then $\E[S_T^m] = + \infty$ for any $T>0$.
        \item[(ii)] If $|\rho| > \sqrt{1 - \frac{1}{m}}$, then $\E[S_T^m] < + \infty$ for any $T>0$.
    \end{enumerate}
\end{theorem}

\begin{proof}
    The proof is given in Section \ref{sect:moments}.
\end{proof}


{ 
Theorem~\ref{thm:moments} provides the threshold $\bar m$ for the finiteness of the price process moments, which determines the asymptotic behavior of the implied volatility smile as $K \to +\infty$ in the signature volatility model \eqref{eq:price_model}, using the celebrated formula by \citet{lee2004moment}.

\begin{corollary}\label{cor:lee} Let $S$ be given by \eqref{eq:price_model}. 
    Suppose that the order $N \geq 2$ is odd and $\rho\sigma^{\word{2}\conpow{N}} < 0$. For a maturity $T>0$ and a strike $K \in \mathbb R$ we denote by $\sigma_{\mathrm{BS}}^2(T, k)$ the implied volatility of vanilla (call or put) options  with $k = \log\left(\frac{K}{S_0}\right)$ the log-moneyness.
Then, the following asymptotic behavior of the implied volatility holds 
\begin{equation}\label{eq:lee_wings}
\limsup_{k \to +\infty}\dfrac{\sigma_{\mathrm{BS}}^2(T, k)T}{k} = \beta_R = 2\dfrac{1 - |\rho|}{1 + |\rho|}.
\end{equation}
\end{corollary}

\begin{proof}
    An application of the \citet[Theorem 3.2]{lee2004moment} formula yields the asymptotic behavior \eqref{eq:lee_wings} with $\beta_R = 2 - 4\left(\sqrt{\bar p^2 + \bar p} - \bar p\right)$, where $\bar p = \sup\left\{p\colon\ \E\left[S_T^{1+p}\right] < \infty\right\}$. Under the assumptions of Theorem~\ref{thm:moments}, a straightforward computation gives 
    \begin{equation*}
    \bar p = \dfrac{\rho^2}{1 - \rho^2}, \qquad \beta_R = 2\dfrac{1 - |\rho|}{1 + |\rho|}.
    \end{equation*}
\end{proof}

To illustrate Corollary~\ref{cor:lee}, we compute the implied volatilities for call options with $T = 1$ by simulating $10^7$ price trajectories in the signature volatility model with $N = 3$. The non-leading coefficients $\bsigma$ are again chosen randomly from the uniform distribution on $[-0.5,, 0.5]$, while the leading coefficient $\sigma^{\word{2}\conpow{N}}$ is fixed at $10$. The simulated total variance smiles $\sigma_{\mathrm{BS}}^2(T, k)T$ for $\rho = -0.7$ and $\rho = -0.8$, together with the asymptotics predicted by Lee’s formula~\eqref{eq:lee_wings}, are shown in Figure~\ref{fig:sig_wings}.
}

\begin{figure}[H]
    \begin{center}
    \includegraphics[width=0.7\linewidth]{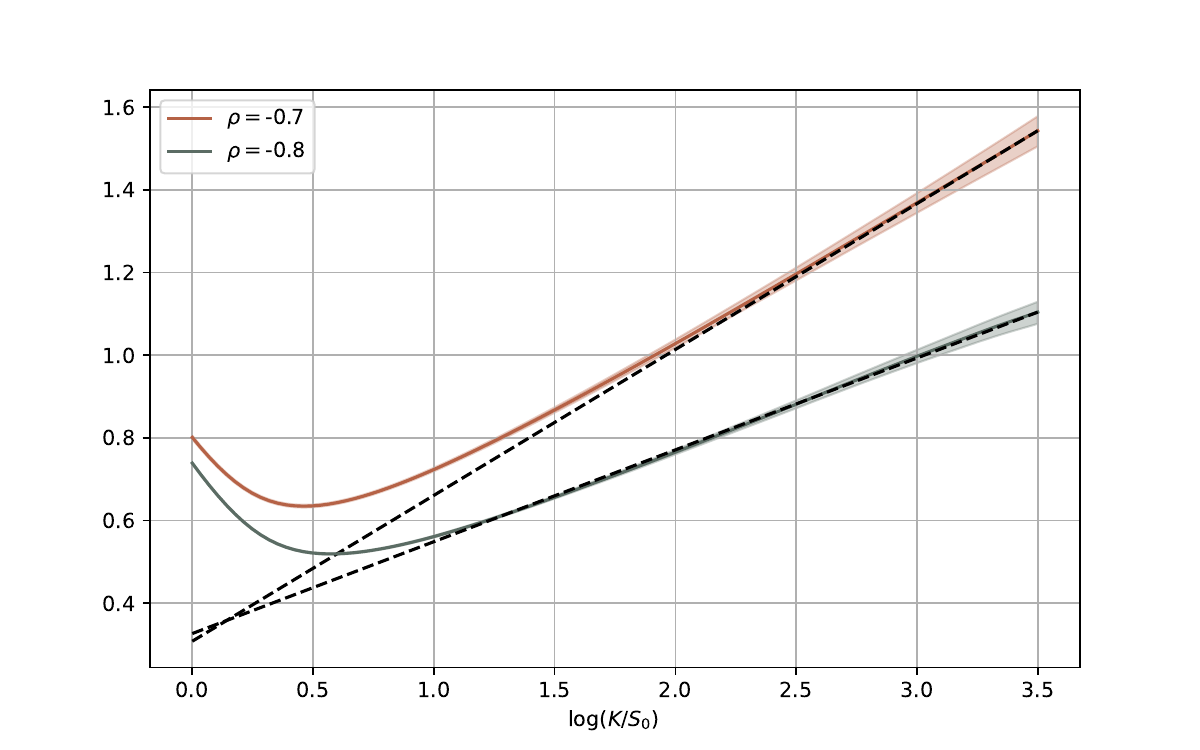}
    \caption{{ Implied total variance $\sigma_{\mathrm{BS}}^2(T, k)T$ for $N = 3$ and correlation parameters $\rho = -0.7$ (brown) and $\rho = -0.8$ (dark green). Shaded regions represent $95\%$ confidence intervals. The dashed black lines show the asymptotic slopes $\beta_R$ as given by~\eqref{eq:lee_wings}.}}
    \label{fig:sig_wings}
    \end{center}
\end{figure}

{
\begin{sqremark}
The proof of Theorem~\ref{thm:moments} crucially relies on the presence of a dominant term that determines the outcome of the moment explosion problem. The sign of this term is that of \( \tfrac{\bar{\rho}^2 m^2 - m}{2} \). In the critical case where \( |\rho| = \sqrt{1 - \tfrac{1}{m}} \), this term vanishes, and the finiteness of the $m$-th moment may or may not hold, depending on the other coefficients of~$\bsigma$ and on the time horizon~$T$, as we illustrate in the proposition below. Hence, we do not expect a simple general criterion to exist.
\end{sqremark}

\begin{proposition}\label{prop:crit}
Let \( \bsigma = \alpha \word{222} + \beta \word{221} \) with \( \alpha \neq 0 \) and \( \beta \in \R \). 
Let \( S \) be given by \eqref{eq:price_model}, and assume that \( \rho \alpha < 0 \). 
Then, for \( |\rho| = \sqrt{1 - \tfrac{1}{m}} \) and \( m > 1 \),
\begin{enumerate}
    \item If \( \beta = 0 \) and \( T^2 < \tfrac{2\pi^2}{\alpha\sqrt{m(m-1)}} \), then \( \E[S_T^m] < +\infty \).
    \item If \( \beta = 0 \) and \( T^2 > \tfrac{2\pi^2}{\alpha\sqrt{m(m-1)}} \), then \( \E[S_T^m] = +\infty \).
    \item If \( \beta \neq 0 \), then \( \E[S_T^m] = +\infty \) for all \( T > 0 \).
\end{enumerate}
\end{proposition}

\begin{proof}
The proof is given in Subsection~\ref{sect:proof_moments_prop}.
\end{proof}
}

\subsection{Extension of the martingality result to multidimensional setting}\label{sect:multid_extension}
Theorem~\ref{T:main_martingality} was stated for the signature volatility model \eqref{eq:price_model}, where the volatility process was constructed based on the signature of a one-dimensional Brownian motion $W$ correlated with $B$. Consider now a model
\begin{equation}
\begin{aligned}\label{eq:price_model_multid}
    \dfrac{dS_t}{S_t} &= \sigma_t dB_t, \quad t \geq 0, \\
    \sigma_t &= \bracketsigX[t][\widehat{Y}]{\bsigma},
\end{aligned}
\end{equation}

where $\sigX[t][\widehat Y]$ denotes the signature of $\widehat Y = (t, B, Z^1, \ldots, Z^d)$, $Z = (Z^1, \ldots, Z^d)$ is a $d$-dimensional Brownian motion independent from $B$ and $\bsigma\in\tTA[d+2]{N}$. Theorem~\ref{T:main_martingality} can be generalized as follows.

\begin{theorem}\label{T:main_martingality_multid}
 Let $S$ be given by \eqref{eq:price_model_multid}.      Suppose that the leading coefficient $\sigma^{\word{2}\conpow{N}} \not= 0$.
     \begin{enumerate}
         \item[(i)] If the order $N = 1$, then $S$ is a martingale. 
         \item[(ii)] If the order $N \geq 2$, then $S$ is a martingale if and only if $N$ is odd and $\sigma^{\word{2}\conpow{N}} \leq 0.$
     \end{enumerate}
\end{theorem}
\begin{proof}
    \textit{(i)} and sufficiency in \textit{(ii)} are proved in Section~\ref{sect:martingality_expec}, while necessity in \textit{(ii)} is proved in Subsection~\ref{section:proof_local_mart_mulitiD}.
\end{proof}

\begin{remark}\label{rmk:multid_model}
    The price model \eqref{eq:price_model} is indeed a special case of the model \eqref{eq:price_model_multid}, since the elements of the signature of $\widehat{W} = (t, W) = (t, \rho B + \sqrt{1 - \rho^2} B^\perp)$ can be written as a linear combination of elements of the signature of $\widehat{Y} = (t, B, B^\perp)$, such that $\bracketsig{\bsigma} = \bracketsigX[t][\widehat{Y}]{\tilde\bsigma}$ for some $\tilde\bsigma \in \tTA[3]{N}$ with the same order $N$. Moreover, the coefficient of $\tilde\bsigma$ corresponding to $\word{2}\conpow{N}$ equals $\rho^N\bsigma^{\word{2}\conpow{N}}$.
\end{remark}

\begin{remark}
    Although Theorem~\ref{T:main_martingality} directly follows from Theorem~\ref{T:main_martingality_multid} when $ \rho \neq 0 $, we chose to present the multidimensional framework after the one-dimensional one for several reasons. First, the model in \eqref{eq:price_model} is more natural and efficient in practice, as it uses a two-dimensional alphabet instead of the three-dimensional one used in its multidimensional counterpart. Moreover, the link between Theorems~\ref{T:main_martingality} and~\ref{thm:moments} and their Markovian counterparts, is much clearer when the one-dimensional model is introduced first. Finally, the methods used in the proof of Theorem~\ref{thm:moments} rely heavily on the structure of the equations in model \eqref{eq:price_model} and would be difficult to generalize to the multidimensional framework.
\end{remark}

\section{A general martingality criterion}\label{section:general_crit}

In this section, we derive a general criterion to establish the martingality of the price process in signature volatility models. Similarly to the Markovian setting, lack of martingality is seen by a Girsanov transform to be equivalent to the explosion of a certain SDE of signature drift form.

We let $S$ be given by \eqref{eq:price_model_multid}, i.e. the volatility process is given by
\begin{equation}
    \sigma_t = \left\langle \bsigma, \widehat{\mathbb{Y}} _t\right\rangle, \quad t \geq 0
\end{equation}
where $\widehat{\mathbb{Y}}_t$ is the signature of the process $\widehat Y_t = (t, B_t, Z^1_t, \ldots, Z^d_t)$, and $\bsigma \in \tTA[d+2]{N}$ for some $N \geq 1$.

The criterion will be stated, in Theorem~\ref{T:explosion_criterion},  in terms of the solution to the SDE
\begin{equation} \label{eq:sdeX}
    X_t =  \int_0^t \left\langle \bsigma, \widehat{\mathbb{Y}}^X_s \right\rangle ds + B_t,
\end{equation}
where $\widehat{\mathbb{Y}}^X_t$ denotes the signature of $\widehat{Y}^X_t = (t, X_t, Z^1_t, \ldots, Z^d_t)$.

\subsection{Signature SDE}

In this subsection, we show the signature SDE \eqref{eq:sdeX}
is well-defined (up to an explosion time) for arbitrary $\bsigma$ in the truncated tensor algebra $\tTA[d+2]{N}$. 

We note that $X$ may explode in finite time, so the solution is understood as a process satisfying equation \eqref{eq:sdeX} on {$[[0,\tau^X_{\infty}[[$, where the explosion time $\tau^X_{\infty}$ is a stopping time defined by
\begin{equation}\label{eq:def_rau_inf_tau_n_X}
    \tau^X_{\infty} := \lim_{n \to\infty} \tau_n^X, \quad \tau_n^X:= \inf\{t \geq 0\colon\, |X_t| = n\},
\end{equation}
such that, on the event $\{\tau^X_{\infty} < \infty\}$, we have $\lim\limits_{t \to \tau^X_{\infty}} |X_t| = +\infty$ a.s.}

\begin{proposition} \label{prop:sigODE}
    The SDE \eqref{eq:sdeX} admits a unique strong solution. In addition, it holds that
    \begin{equation} \label{eq:explT}
        \tau^X_{\infty} = \inf \left\{ t \geq 0 , \;\; \int_0^t \langle\bsigma, \widehat{\mathbb{Y}}^X_s \rangle^2 ds = + \infty\right\}.
    \end{equation}
\end{proposition}

\begin{proof}
    For the well-posedness of the SDE, we note that it can be equivalently written as a system of classical (Markovian) SDEs. Indeed, setting $\widehat{\mathbb{Y}}^{X, \word{w}}_t = \langle\word{w}, \widehat{\mathbb{Y}}^X_t \rangle$, the signature SDE \eqref{eq:sdeX} is equivalent to $X = \widehat{\mathbb{Y}}^{X, \word{2}}$, where the $(\widehat{\mathbb{Y}}^{X, \word{w}})_{\word{w}:\, |\word{w}| \leq N}$ satisfy $\widehat{\mathbb{Y}}^{X, \emptyword} \equiv 1$ and for $|\word{u}| < N$, 
    \begin{align*}
    \widehat{\mathbb{Y}}^{X, \word{u}\word{1}}_t =& \int_0^t \widehat{\mathbb{Y}}^{X, \word{u}}_s ds, \\  
    \widehat{\mathbb{Y}}^{X, \word{u}\word{2}}_t =& \int_0^t \widehat{\mathbb{Y}}^{X, \word{u}}_s \left( \sum_{|\word{w}|\leq N} \bsigma^{\word{w}} \widehat{\mathbb{Y}}^{X, \word{w}}_s \right) ds + \int_0^t \widehat{\mathbb{Y}}^{X, \word{u}}_s \circ dB_s, \\
    \widehat{\mathbb{Y}}^{X, \word{u}\word{i}}_t =& \int_0^t \widehat{\mathbb{Y}}^{X, \word{u}}_s {\circ} dZ^{i - 2}_s,  \;\; i=3,\ldots, d + 2.\\
    \end{align*}
    This is a Stratonovich SDE with locally Lipschitz coefficients, for which existence and uniqueness up to an explosion time are standard.

    It remains to prove \eqref{eq:explT}. 
    {Since on the event $\{\tau^X_{\infty} < \infty\}$, we have $\int_0^{\tau^X_{\infty}}\langle\bsigma, \widehat{\mathbb{Y}}^X_s \rangle\, ds = X_{\tau^X_{\infty}} - B_{\tau^X_{\infty}} = \pm\infty$, it follows that
    \begin{equation*}
        \int_0^{\tau^X_{\infty}}\langle\bsigma, \widehat{\mathbb{Y}}^X_s \rangle^2 ds \geq \left(\int_0^{\tau^X_{\infty}}\langle\bsigma, \widehat{\mathbb{Y}}^X_s \rangle\, ds\right)^2 = +\infty,
    \end{equation*}
    we have  $\int_0^{\tau^X_{\infty}}\langle\bsigma, \widehat{\mathbb{Y}}^X_s \rangle^2 ds = +\infty $ on $\{\tau^X_{\infty} < \infty\}$}. Conversely, by Lemma~\ref{lem:expYV} {(stated in Section~\ref{sect:bounds_proofs} below)}, we can show that, for each word~$\word{w}$,
    \[ \E\left[\int_0^{\tau_n^X\wedge T} \langle \word{w}, \widehat{\mathbb{Y}}^X_t \rangle^2 dt\right]\leq { C_{\word{w},T} n^{2|\word{w}|+1}} < \infty.  \]
    
    Since $\tau^X_{\infty} = \sup\limits_{n \geq 1} \tau^X_n$, this implies that $\int_0^t \langle\bsigma, \widehat{\mathbb{Y}}^X_s \rangle^2 ds < \infty$ for each $t<\tau^X_{\infty}$.
\end{proof}
{
}

\subsection{Martingality and explosion}
By the results of the previous subsection, equation \eqref{eq:sdeX} admits a unique strong solution up to an explosion time $\tau^X_{\infty}$, satisfying \eqref{eq:explT}.

\begin{theorem}\label{T:explosion_criterion}
For each $T>0$, for $S$ defined by \eqref{eq:price_model_multid}, it holds that
\[
S \mbox{ is a martingale on $[0,T]$} \; \Leftrightarrow \P(\tau^X_{\infty} < T ) = 0,
\]
where $\tau^X_{\infty}$ is the explosion time of $X$ defined by \eqref{eq:sdeX}.
\end{theorem}

 \begin{proof}
Since $S$ is a nonnegative local martingale, it is a supermartingale by Fatou's Lemma, so that $\E[S_T] \leq S_0$ for all $T \geq 0$. Hence, for a given $T>0$, it is a martingale on $[0,T]$ if and only if $\E\left[S_T\right] = S_0$. Let
\[ 
\tau_n = \inf\left\{ t > 0\colon\, \int_0^t \left\langle \bsigma, \widehat{\mathbb{Y}}_s\right\rangle^2 ds =n \right\}.
\]
As the stopped process $\left(\int_0^{t\land\tau_n} \left\langle \bsigma, \widehat{\mathbb{Y}}_s\right\rangle^2 ds\right)_{t \geq 0}$ is bounded by $n$, the stopped price process $S^{\tau_n} = (S_{t \land \tau_n})_{t \geq 0}$ is an exponential martingale by Novikov's criterion and it follows that
$$
S_0 = \E\left[S^{\tau_n}_T \right] = \E \left[S_{T\land\tau_n}\right] = \E\left[S_T\indic{T \leq \tau_n}\right] + \E \left[S_{\tau_n}\indic{\tau_n < T}\right] = \E\left[S_T\indic{T \leq \tau_n}\right] + S_0\widehat{\P}_n\left({\tau_n < T}\right),
$$
where $\dfrac{d\widehat{\P}_n}{d\P} = \dfrac{S_{T\land \tau_n}}{S_0}$.

By the monotone convergence theorem, the term $\E\left[S_T\indic{T \leq \tau_n}\right]$ converges to $\E\left[S_T\right]$ as $n \to \infty$. Thus,
$$
S_0 - \E\left[S_T\right] = \lim_{n \to\infty} \E \left[S_{\tau_n}\indic{\tau_n < T}\right] = \lim_{n \to\infty} S_0\widehat{\P}_n\left({\tau_n < T}\right).
$$
On the other hand, by Girsanov theorem, for each $n$, it holds that for each $t \leq T$,
\[
B_t = \int_{0}^{t \wedge \tau_n} \left\langle \bsigma, \widehat{\mathbb{Y}}_s\right\rangle ds + \widehat{B}_t,
\]
where $\widehat{B}$ is a Brownian motion under $\widehat{\P}_n$. { In particular, the Girsanov change-of-measure argument implies the uniqueness of the weak solution to the signature SDE (satisfied by \( B \) under \( \widehat{\P}_n \)):
\begin{equation}\label{eq:sig-sde_stopped}
    X_t^n 
    = \int_0^t 
        \left\langle \bsigma, \widehat{\mathbb{Y}}^{X^n}_s \right\rangle 
        \indic{\int_0^s \left\langle \bsigma, \widehat{\mathbb{Y}}^{X^n}_r \right\rangle^2 dr \le n} 
        ds 
      + B_t.
\end{equation}
Indeed, by applying the Girsanov theorem as above, and noting that the corresponding Radon–Nikodym derivative is a true martingale, one verifies that the law of $X^n$ can be expressed in terms of Wiener measure and is therefore uniquely determined. It remains to observe that \eqref{eq:sig-sde_stopped} is also satisfied by
\[
X_t^n = X_{t \land \tilde{\tau}_n} + \indic{t \ge \tilde{\tau}_n}(B_t - B_{\tilde{\tau}_n}),
\]
where \( X \) denotes the strong solution to \eqref{eq:sdeX} given by Proposition~\ref{prop:sigODE}, so that
\[
\tilde{\tau}_n 
= \inf\left\{ t > 0 : \int_0^t \left\langle \bsigma, \widehat{\mathbb{Y}}^{X}_s \right\rangle^2 ds = n \right\} 
= \inf\left\{ t > 0 : \int_0^t \left\langle \bsigma, \widehat{\mathbb{Y}}^{X^n}_s \right\rangle^2 ds = n \right\}.
\]
Hence, weak uniqueness implies that
\[
\widehat{\P}_n\left(\tau_n < T\right) = \P\left(\tilde{\tau}_n < T\right),
\]
and therefore,
\[
\lim_{n \to \infty} \widehat{\P}_n\left(\tau_n < T\right)
= \lim_{n \to \infty} \P\left(\tilde{\tau}_n < T\right)
= \P\left(\tau^X_{\infty} < T\right),
\]
where we have used that \( \lim_{n \to \infty} \tilde{\tau}_n = \tau^X_{\infty} \).
}
 \end{proof}

\begin{corollary}\label{cor:one_d_criterion}
    The same proof applied to the one-dimensional model implies that $S$ given by \eqref{eq:price_model} is a martingale on $[0, T]$ if and only if $\P(\tau^X_{\infty} < T) { = 0}$, where $\tau^X_{\infty}$ is the explosion time of $X$ defined by
    \begin{equation}\label{eq:sig_SDE_1D}
        X_t =  \rho\int_0^t \left\langle \bsigma, \widehat{\mathbb{X}}_s \right\rangle ds + W_t,
    \end{equation}
    and $\widehat{\mathbb{X}}_t$ denotes the signature of $\widehat{X}_t = (t, X_t)$, for $t \in [0, T]$.
\end{corollary}
 
\begin{remark}
    Theorem \ref{T:explosion_criterion} is closely related to the result proved in \cite[Lemma 1.5]{kazamaki2006continuous} rewritten in terms of new probability measures. A more general criterion was proposed in \cite[Theorem 3.3]{Ruf_2015}. It states that the price process is a martingale if and only if
    \begin{equation}
        \widehat\P\left(\int_0^{T \land \vartheta}\sigma_s^2 ds < \infty\right) = 1, \quad T\geq 0,
    \end{equation}
    where $\vartheta$ stands for the explosion time and $\widehat{\P}$ denotes the extension of $(\widehat\P_n)_{n \geq 1}$ which is constructed in \cite[Proposition 3.1]{Ruf_2015}. However, this approach assumes a specific structure of the underlying probability space as highlighted in \cite[Example 4.3]{Ruf_2015}. For this reason, we prefer to work with Theorem~\ref{T:explosion_criterion}, ensuring that our results are independent of the choice of probability space.
\end{remark}

{
\section{Key lemmas for bounds on the signatures}\label{sect:bounds_proofs}

In this section, we establish bounds on signature elements derived from Lyndon words, which will be essential in the subsequent analysis. The first result states that terms in the signature of $\widehat X_t = (t,X_t)$ may be bounded in terms of the supremum norm of $X$.

\begin{lemma}\label{lemma:linear_from_upper_bound}
    Let $X\colon [0,T] \to \R$ be a continuous path, and let $\widehat{\mathbb{X}}_t$ denote the signature of  $\widehat{X}$.    
    Let $N_{\word{2}}(\word{w})$ denote the number of occurrences of letter $\word{2}$ in $\word{w}$. Then, for any word $\word{w}$, there exists a constant $C_{\word{w}} > 0$, such that
    \begin{equation}\label{eq:linear_form_as_bound}
     \sup_{t \in [0, T]}\left|\bracketsigX[t][\widehat{X}]{\word{w}} \right| \leq C_{\word{w}}T^{|\word{w}|-N_{\word{2}}(\word{w})}\max_{t \in [0, T]}\left|X_t\right|^{N_{\word{2}}(\word{w})}.
    \end{equation}
\end{lemma}
\begin{proof}
    The proof is given in Subsection~\ref{sect:first_lemma_proof}.
\end{proof}

We then consider the signature of the extended process $\widehat Y^X_t = (t, X_t, Z^1_t, \ldots, Z^d_t)$, where $Z_t = (Z^1_t, \ldots, Z^d_t)$ is a $d$-dimensional Brownian motion and $X$ is an arbitrary Itô process. While the terms of the signature no longer satisfy an almost sure bound as in the previous lemma, we show that a similar result still holds in a probabilistic sense.

\begin{lemma} \label{lem:expYV}
Let $Z = (Z^{1},\ldots,Z^d)$ be a $d$-dimensional Brownian motion, and let $X$ be an Itô process defined on the same probability space and such that $\langle Z^i, X\rangle \equiv 0$ for any $i = 1, \ldots, d$.
Let $\widehat{\mathbb{Y}}^X_t$ denote the signature of $\widehat Y_t^X$.  
Then, for any word $\word{w}$, any $q \geq 2$, $\epsilon > 0$, and any stopping time $\tau$ taking values in $[0,T]$, there exists a constant $C = C_{\word{w}, q, \epsilon, T}$ such that  
\begin{equation}\label{eq:second_bound}
\E\left[\int_0^{\tau} \left\langle \word{w}, \widehat{\mathbb{Y}}^X_t \right\rangle^{q} dt\right] \leq C \cdot \E \left[\int_0^{\tau} X_t^{q (N_{\word{2}}(\word{w}) + \epsilon)} dt\right],
\end{equation}
where $N_{\word{2}}(\word{w})$ denotes the number of occurrences of the letter $\word{2}$ in $\word{w}$.
\end{lemma}
\begin{proof}
    The proof is given in Subsection~\ref{sect:second_lemma_proof}.
\end{proof}
{\begin{remark}
    The statement of Lemma~\ref{lem:expYV} remains true if, for some $C_0 > 0$ and for any $i = 1, \ldots, d$,  
    \begin{equation*}\label{eq:covariation_assumption}
        \left| d \langle Z^i, X \rangle_t \right| \leq C_0 \, dt.
    \end{equation*}  
    In this case, the same proof is applicable since the cross-variation term appearing in the bounds can be controlled in exactly the same way as the drift term.
\end{remark}
}
}
\subsection{Polynomial representations in the shuffle algebra}\label{sect:polynomial_repr}

In this subsection, we recall the notion of \cite{Lyndon1954}  words which will help us prove the bounds given by Lemma~\ref{lemma:linear_from_upper_bound} and Lemma~\ref{lem:expYV}.
The extended tensor algebra $\eTA$ can be considered as an algebra endowed with a shuffle product. We define a shuffle multivariate polynomial $P\shupow{}$ of $n$ arguments as
$$
P\shupow{}: \eTA^n \to \eTA, \quad (\word{v_1}, \ldots, \word{v_n}) \mapsto 
 \sum_{\substack{p_1, \ldots, p_n}} c_{p_1, \ldots, p_n} \, \word{v_1}\shupow{p_1}\shuprod\word{v_2}\shupow{p_2}\ldots\shuprod\word{v_n}\shupow{p_n},
$$
where the sum is finite and taken over some set of positive natural numbers $p_1, \ldots, p_n$.

The Lyndon words defined below are useful to construct the polynomial representations of the shuffle algebra elements.

\begin{definition}
    A Lyndon  word in an alphabet $\alphabet$ is a non-empty word that is strictly lexicographically greater than any of its non-trivial cyclic rotations. A set of all Lyndon words in the alphabet $\alphabet$ will be denoted by $\mathcal{L}(\alphabet)$, and $\mathcal{L}_N(\alphabet)$ will denote the Lyndon words of length not greater than $N$.
\end{definition}

\begin{example}
    The first Lyndon words in the alphabets $\alphabet[2]$ and $\alphabet[3]$ are given by
    \begin{itemize}
        \item $\mathcal{L}(\alphabet[2]) = \{\word{2}, \word{1}, \word{21}, \word{221}, \word{211}, \word{2221}, \word{2211}, \word{2111}, \ldots\}.$
        \item $\mathcal{L}(\alphabet[3]) = \{\word{3}, \word{2}, \word{1}, \word{32}, \word{31}, \word{21}, \word{332}, \word{331}, \word{322}, \word{321}, \word{312}, \word{311}, \word{221}, \word{211}, \ldots\}.$
    \end{itemize}
\end{example}

An important property of the Lyndon words is given by the  \cite{Radford1979ANR} Theorem. It states that the Lyndon words generate the shuffle algebra, i.e.~each element $\bell \in \TA$ can be represented by a shuffle polynomial $P_{\bell}\shupow{}$:
\begin{equation}
    \bell = P_{\bell}\shupow{}(\word{v_1}, \ldots, \word{v_{n_{\bell}}}), \quad \word{v_1}, \ldots, \word{v_{n_{\bell}}} \in \mathcal{L}.
\end{equation}
Moreover, if $\bell \in \tTA{N}$, then $\word{v_1}, \ldots, \word{v_{n_{\bell}}} \in \mathcal{L}_N$ and $\sum_{i = 1}^{n_{\bell}}p_i|\word{v_i}| \leq N$ {for all terms $\word{v_1}\shupow{p_1}\shuprod\ldots\shuprod\word{v_n}\shupow{p_n}$ in the polynomial $P_{\bell}\shupow{}$.}

The following example provides explicit polynomial representations for {the words over alphabet $\{\word{1},\word{2}\}$} of length up to $3$ that can be easily verified.

\begin{example}
    \begin{enumerate}
        \item $\mathcal{L}_1 = \{\word{1}, \word{2}\}$ contains all the words of length $1$.
        \item $\mathcal{L}_2 = \{\word{1}, \word{2}, \word{21}\}$. The remaining words of length $2$ can be obtained as the following polynomials:
        \begin{align}
            \word{11} &= \dfrac{1}{2}\cdot\word{1}\shupow{2}, \\
            \word{12} &= \word{1}\shuprod\word{2} - \word{21}, \\
            \word{22} &= \dfrac{1}{2}\cdot\word{2}\shupow{2}.
        \end{align}
        \item $\mathcal{L}_3 = \{\word{1}, \word{2}, \word{21}, \word{211}, \word{221}\}$. The words of length $2$ were constructed above.  The remaining words of length $3$ can be obtained as polynomials:
    \begin{equation}
    \begin{aligned}
        \word{111} &= \dfrac{1}{6} \cdot\word{1}\shupow{3},  
        &\word{122} &= \dfrac{1}{2}\cdot\word{1}\shuprod\word{2}\shupow{2} - \word{21}\shuprod\word{2} + \word{221}, \\
        \word{112} &= \word{11}\shuprod\word{2} - \word{21}\shuprod\word{1} + \word{211}, 
        &\word{212} &= \word{21}\shuprod\word{2} - 2\cdot\word{221}, \\
        \word{121} &= \word{21}\shuprod\word{1} - 2\cdot\word{211}, 
        &\word{222} &= \dfrac{1}{6}\cdot\word{2}\shupow{3}.
    \end{aligned}
    \end{equation}
    \end{enumerate}
\end{example}
We note that no Lyndon word ends with the lexicographically greatest letter $\word{d}$, except for the word $\word{d}$ itself. Together with Radford's theorem, and noting that the letters in the alphabet can be ordered arbitrarily, this implies the following result.

\begin{proposition} \label{prop:nod}
Let $\word{w}$ be a word in $A_d$, and fix $\word{k} \in A_d$. Then $\word{w}$ can be expressed as a shuffle polynomial  
\[
\word{w} = P^{\shuprod}_{\word{w}}(\word{v_1}, \ldots, \word{v_k}),
\]
where each $\word{v_i}$ either does not end with the letter $\word{k}$ or is equal to $\word{k}$.
\end{proposition}

\begin{remark}
    {In the subsequent proofs, we will often deal with the signatures of \( \widehat{Y}^X_t = (t, X_t, Z_t^1, \ldots, Z_t^d) \), where \( X_t \) is an arbitrary Itô process and \( Z_t \) is a \( d \)-dimensional Brownian motion. It is then beneficial, in proofs using induction on word length, to represent an arbitrary word \( \word{w} \) as a polynomial of words that do not end with \( \word{2} \) (except for \( \word{2} \conpow{k} \)), corresponding to the process \( X \). This representation reduces the analysis to standard Riemann integrals and Stratonovich integrals with respect to the Brownian motion.}
\end{remark}


\subsection{Proof of Lemma~\ref{lemma:linear_from_upper_bound}}\label{sect:first_lemma_proof}

We will prove the result by induction on the length $ |\word{w}| $ of the word $\word w$.  
For $ |\word{w}| = 1 $, the result holds trivially with $ C_{\word{w}} = 1 $.  For $ |\word{w}| > 1 $, {we assume that \eqref{eq:linear_form_as_bound} holds for words of length strictly less than $|\word{w}|$.}  
By Proposition~\ref{prop:nod}, the word $ \word{w} $ can be written as a linear combination of monomials of the form  
\begin{equation}
    \word{2}\shupow{k}\shuprod\word{u_11}\shuprod\ldots\shuprod\word{u_m1},
\end{equation}
where 
\begin{equation}\label{eq:lengths_eq}
    k + \sum_{i=1}^mN_{\word{2}}(\word{u_i}) = N_{\word{2}}(\word{w}), \quad  k + m + \sum_{i=1}^m|\word{u_i}|= |\word{w}|.
\end{equation}
Hence, it suffices to bound each monomial:
\begin{align}
\left|\bracketsigX[t][\widehat{X}]{\word{2}\shupow{k}\shuprod\word{u_11}\shuprod\ldots\shuprod\word{u_m1}}\right| &= \left|X_t^k\prod_{i=1}^m\int_0^t\bracketsigX[s][\widehat{X}]{\word{u_i}}ds\right| \\
    &\leq \max_{t \in [0, T]}|X_t|^k\prod_{i=1}^m \left(T \max_{t \in [0, T]} \left|\bracketsigX[t][\widehat{X}]{\word{u_i}}\right|\right) \\
    &\leq T^m\max_{t \in [0, T]}|X_t|^k \prod_{i=1}^m\left(C_{\word{u_i}}T^{|\word{u_i}| - N_{\word{2}}(\word{u_i})}\max_{t \in [0, T]}|X_t|^{N_{\word{2}}(\word{u_i})}\right) \\
    &=  \left(\prod_{i=1}^mC_{\word{u_i}}\right) T^{m + \sum\limits_{i=1}^m\left(|\word{u_i}| - N_{\word{2}}(\word{u_i})\right)}\max_{t \in [0, T]}|X_t|^{k + \sum\limits_{i=1}^{m}N_{\word{2}}(\word{u_i})},
\end{align}
where we have used for the second inequality the induction hypothesis \eqref{eq:linear_form_as_bound} on the words $\word{u_i}$,  since $|\word{u_i}|<|\word w|$. It remains to observe that, by
\eqref{eq:lengths_eq},
\begin{equation}
    k + \sum\limits_{i=1}^{m}N_{\word{2}}(\word{u_i}) = N_{\word{2}}(\word{w}), \quad
    m + \sum\limits_{i=1}^m\left(|\word{u_i}| - N_{\word{2}}(\word{u_i})\right) = |\word{w}| - N_{\word{2}}(\word{w}),
\end{equation}
which yields \eqref{eq:linear_form_as_bound} for the word $\word w$ and ends the proof.

\subsection{Proof of Lemma~\ref{lem:expYV}}\label{sect:second_lemma_proof}
We proceed by induction on the length $n$ of the word $\word{w}$. For $n=1$, \eqref{eq:second_bound} holds trivially. For $n\geq 2$, we assume the  bound to hold for all words of length strictly less than $n$.

If $\word{w} = \word{2}^{\otimes n}$, the claim is obvious, as equality holds in \eqref{eq:second_bound} for {$ C = (n!)^{-q}$} and $\epsilon=0$.

Otherwise, by Proposition \ref{prop:nod}, one can write $\word{w}$ as
\[
\word{w} = \sum_{k=1}^{m} \alpha_k \word{w_{k,1}} \shuprod \ldots \shuprod \word{w_{k,i_k}}
\]
where the words $\word{w_{k,j}}$ contain at least one letter and satisfy
\[
|\word{w_{k,1}}|+ \ldots + |\word{w_{k,i_k}}| = |\word{w}| = n,
\]
\[
N_{\word{2}}(\word{w_{k,1}})+ \ldots + N_{\word{2}}(\word{w_{k,i_k}}) = N_{\word{2}}(\word{w}),
\]
and for each $k,j$, $\word{w_{k,j}}$ is either of length not greater than $n-1$ or ending with a letter different from $\word{2}$.

By the elementary inequality  
$$
(x_1 + \ldots + x_m)^q \leq C_{q,m} \left(|x_1|^q + \ldots + |x_m|^q\right),
$$  
it suffices to bound each term corresponding to the summands separately.  

We then distinguish two cases. First, assume that $i_k \geq 2$, meaning that each of the words has length less than $n$. In this case, we apply Young's inequality to write  
$$
  \prod_{j=1}^{i_k} \left|\left\langle\word{w_{k,j}}, \widehat{\mathbb{Y}}^X_t\right\rangle\right|^{q}  \lesssim \sum_{j=1}^k   \left|\left\langle\word{w_{k,j}}, \widehat{\mathbb{Y}}^X_t\right\rangle\right|^{q \frac{N_{\word{2}}(\word{w}) + \epsilon}{N_{\word{2}}(\word{w_{k,j}})+\epsilon_j}},
$$
where $\epsilon_j >0$ are arbitrary such that $\sum_j \epsilon_j = \epsilon$. We then conclude by the induction hypothesis applied to each $\word{w_{k,j}}$ since $|\word{w_{k,j}}| < n$.

In the case where $i_k=1$, we know that $\word{w_{k,1}}$ is of the form $\word{u}\word{1}$ or $\word{u}\word{j}$ for $3\leq j \leq d + 2$ . In the first case, this means that
\begin{align*}
    \E \left[ \int_0^{\tau} \left|\left\langle\word{w_{k,1}}, \widehat{\mathbb{Y}}^X_t\right\rangle\right|^{q} dt\right] = 
    \E\left[ \int_0^{\tau}\left|\int_0^t \left\langle\word{u}, \widehat{\mathbb{Y}}^X_s\right\rangle ds\right|^{q} dt \right] 
    \leq T^{ q - 1} \E \left[\int_0^{\tau} \left| \left\langle\word{u}, \widehat{\mathbb{Y}}^X_s\right\rangle\right|^{q} ds  \right],
\end{align*}
where we have used H\"older's inequality, and we conclude by the induction hypothesis applied to $\word{u}$ since $|\word{u}| = n - 1$.

For the case $\word{w_{k, 1}}=\word{u}\word{j} = \word{u'u_nj}$, $j \geq 3$, this means that
\begin{align*}
\left\langle\word{w_{k,1}}, \widehat{\mathbb{Y}}^X_t\right\rangle &= \int_0^t \left\langle\word{u}, \widehat{\mathbb{Y}}^X_s\right\rangle \circ dZ^{j - 2}_s \\
&= \int_0^t \left\langle\word{u}, \widehat{\mathbb{Y}}^X_s\right\rangle dZ^{j - 2}_s + \indic{\word{u_n} = \word{j + 2}}\int_0^t \left\langle\word{u'}, \widehat{\mathbb{Y}}^X_s\right\rangle ds + \indic{\word{u_n} = \word{2}}\int_0^t \left\langle\word{u'}, \widehat{\mathbb{Y}}^X_s\right\rangle d\langle X_s, Z^{j - 2}_s\rangle.
\end{align*}
Again, it is enough to bound separately the integrals of each of the three terms. The second one follows exactly as in the case $\word{w_{k,1}}=\word{u}\word{1}$, and the third one vanishes, as follows {from the assumption $\langle Z^i, X\rangle \equiv 0$}. For the first term, using the Burkholder--Davis--Gundy inequality, we obtain
\begin{align*}
     \E \left[\int_0^{\tau} \left|\int_0^t \left\langle\word{u}, \widehat{\mathbb{Y}}^X_s\right\rangle dZ^{j - 2}_s\right|^{q} dt \right] &\leq T  \E \left[\sup_{0\leq t \leq \tau}  \left|\int_0^t \left\langle\word{u}, \widehat{\mathbb{Y}}^X_s\right\rangle dZ^{j - 2}_s\right|^{q}\right] \\
     &\leq C_{T,q} \E\left[\left(\int_0^{\tau} \left\langle\word{u}, \widehat{\mathbb{Y}}^X_s\right\rangle^2 ds \right)^{\frac{q}{2}} \right]\\
    & \leq C'_{T,q} \E \left[ \int_0^{\tau} \left|\left\langle\word{u}, 
     \widehat{\mathbb{Y}}^X_s\right\rangle\right|^q ds  \right]
\end{align*}
Using Hölder's inequality, recalling that $q \geq 2$, and applying the induction hypothesis to $\word{u}$ since $|\word{u}| = n - 1$, we establish \eqref{eq:second_bound} for the word $\word{w}$, thereby completing the proof.

\section{Proof of martingality in Theorem~\ref{T:main_martingality} and Theorem~\ref{T:main_martingality_multid}}\label{sect:martingality_expec}

As noted in Remark~\ref{rmk:multid_model}, Theorem~\ref{T:main_martingality} follows immediately from Theorem~\ref{T:main_martingality_multid} when $\rho \neq 0$.
The case $\rho = 0$ is trivial, as $X$, given by \eqref{eq:sig_SDE_1D}, is a Brownian motion. Consequently, we have $\tau_\infty^X = +\infty$ almost surely, and $S$ is a martingale by Corollary~\ref{cor:one_d_criterion}.
Hence, we turn directly to the more general model \eqref{eq:price_model_multid} and proceed with the proof of Theorem~\ref{T:main_martingality_multid}.

If $N = 1$, the equation \eqref{eq:sdeX} is a linear SDE, which admits a unique strong solution on $\mathbb{R}^+$. This ensures that $\tau_\infty^X = +\infty$ almost surely. We conclude that $S$ is a martingale by Theorem~\ref{T:explosion_criterion}. This establishes \textit{(i)} of Theorem~\ref{T:main_martingality_multid}.

We now consider the non-trivial case $N \geq 2$, corresponding to point \textit{(ii)} of Theorem~\ref{T:main_martingality_multid}.
By Theorem~\ref{T:explosion_criterion}, we need to prove that the solution $X$ to \eqref{eq:sdeX} does not explode in finite time. Specifically, defining
\[
\tau^X_n = \inf \left\{ t \geq 0{ \colon}\, |X_t| = n \right\},
\]
we will show that 
\[
\lim\limits_{n \to \infty}\mathbb{P}(\tau^X_n < T) = \P(\tau^X_\infty < T)  =  0.
\]

The probabilities on the left-hand side are controlled by
\[
\P(\tau^X_n \leq T) \leq \dfrac{1}{n^{2N}}\E\left[X_{T\wedge \tau_n^X}^{2N}\right],
\]
since
$$
\E\left[X_{T\wedge \tau_n^X}^{2N}\right] \geq \E\left[n^{2N}\indic{\tau_n^X \leq T}\right] = n^{2N}\P(\tau^X_n \leq T),
$$
so that it is enough to obtain a uniform bound (in $n$)  for $\E\left[X_{T\wedge \tau_n^X}^{2N}\right]$.

By Itô's formula, we have on $t \in [[0, \tau_n^X]]$
\begin{equation}
    dX_t^{2N} = 2NX_t^{2N-1}dX_t + N(2N-1)X_t^{2N-2}dt = 2NX_t^{2N-1}\bracketsigX[t][\widehat{Y}^X]{\bsigma}dt + 2NX_t^{2N-1}dB_t+ N(2N-1)X_t^{2N-2}dt
\end{equation}
Integration from $0$ to $T\wedge \tau_n^X$, taking the expectation and denoting $\tilde\bsigma = \bsigma - \sigma^{\word{2}\conpow{N}} \word{2}\conpow{N}$, we obtain
\begin{align}
    \E\left[X_{T\wedge \tau_n^X}^{2N}\right] &=  \E \left[\int_0^{T\wedge \tau_n^X} \left(2NX_s^{2N-1}\bracketsigX[s][\widehat{Y}^X]{\bsigma} + N(2N-1)X_s^{2N-2}\right)\,ds \right]\\
    &= \E \left[\int_0^{T\wedge \tau_n^X}\left(\dfrac{2N}{N!} \sigma^{\word{2}\conpow{N}} X_s^{3N-1} + \dfrac{2N}{N!} X_s^{2N-1} \bracketsigX[s][\widehat{Y}^X]{\tilde\bsigma}+N(2N-1)X_s^{2N-2}\right)\,ds \right] \\
    &\leq \E \left[\int_0^{T\wedge \tau_n^X} \left( \dfrac{2N}{N!} \sigma^{\word{2}\conpow{N}} X_s^{3N-1} + \delta X_s^{3N-1} + C(N, \delta) \bracketsigX[s][\widehat{Y}^X]{\tilde\bsigma}^{\frac{3N - 1}{N}}+N(2N-1)X_s^{2N-2}\right]\right)\,ds \\
    &\leq \E \left[\int_0^{T\wedge \tau_n^X} \left({ \left(\dfrac{2N}{N!}\sigma^{\word{2}\conpow{N}} + \delta\right)} X_s^{3N-1} + \tilde C(N, \delta, T) \left(|X_s|^{(3N - 1)\frac{N - 1}{N} + \epsilon} + 1\right)+N(2N-1)X_s^{2N-2}\right)\,ds\right],
\end{align}
where we have applied Young's inequality to obtain the first inequality and Lemma \ref{lem:expYV} for the second one. Choosing $\delta \in \left(0,  -\dfrac{2N}{N!}\sigma^{\word{2}\conpow{N}}\right)$ and $\epsilon$ such that ${(3N - 1)\frac{N - 1}{N} + \epsilon} < 3N - 1$, we ensure that the integrand is bounded from above uniformly in $n$ since $3N - 1$ is even and $\left(\dfrac{2N}{N!}+ \delta\right)\sigma^{\word{2}\conpow{N}} < 0$. This finishes the proof.

\section{Proofs of strict local martingality}\label{section:proof_local_martingality}

In this section, we prove the necessity in Theorem~\ref{T:main_martingality} \textit{(ii)} and Theorem~\ref{T:main_martingality_multid} \textit{(ii)}. Unlike the ``one-shot'' proof of the martingale property in Section~\ref{sect:martingality_expec}, we first prove strict local martingality in the simpler one-dimensional case and then extend it to the multidimensional setting.

\subsection{Proof of necessity in Theorem~\ref{T:main_martingality} \textit{(ii)}}\label{section:proof_local_mart_1D}

The proof of the theorem follows the ideas used by \cite{Lions2007} in the Markovian setting. 

    Fix $T > 0$. According to Corollary~\ref{cor:one_d_criterion}, it is enough to prove than $\P(\tau_\infty^X < T) > 0$, where
    \begin{equation}
        \tau_\infty^X = \lim_{n\to\infty}\tau_n^X, \quad \tau_n^X := \inf\{t \geq 0\colon  |X_t| = n\},
    \end{equation}
    and $X = (X_t)_{t\geq 0}$ satisfies
    \begin{equation}
        dX_t = d W_t +\bracketsig[t][X]{\rho \bsigma}dt, \quad t \in [[0, \tau_\infty^X]], \quad X_0 = 0.
    \end{equation}
    
    For the linear form of signature, the following inequality is a direct consequence of Lemma~\ref{lemma:linear_from_upper_bound}:
    \begin{equation}\label{eq:vol_remainder_ineq}
        \sup_{t \in [0,\, T]}\left| \bracketsigX[t][\widehat{X}]{\bsigma - \sigma^{\word{2}\conpow{N}}\word{2}\conpow{N}} \right| \leq C_{\bsigma, T}\left(1 + \sup_{t \in [0,\, T]}\left|X_t\right|^{N-1}\right) \quad a.s.
    \end{equation}
    for some positive $C_{\bsigma, T} > 0$ depending only on $T$ and the coefficients $\bsigma$. This implies
    \begin{equation}\label{eq:vol_domination_ineq}
        \dfrac{\sigma^{\word{2}\conpow{N}}}{N!}X_t^N - C_{\bsigma, T}\left(1 + \sup_{s \in [0,\, t]}\left|X_s\right|^{N-1}\right) \leq \bracketsigX[t][\widehat{X}]{\bsigma} \leq \dfrac{\sigma^{\word{2}\conpow{N}}}{N!}X_t^N + C_{\bsigma, T}\left(1 + \sup_{s \in [0,\, t]}\left|X_s\right|^{N-1}\right) \quad a.s.
    \end{equation}
    for all $t \in [0,\, \tau^X_\infty]$.

    In particular, it follows that { $\lim\limits_{t \to \tau_{\infty}^X}\bracketsig[t][X]{\bsigma}^2 = +\infty$ on $\{\tau_{\infty}^X < \infty\}$}, so it is sufficient to show that $\P(\tau_\infty^X < T) > 0$.
    
    For an arbitrary $\lambda \in \R$, define the event
    \begin{equation}
    {
    \mathcal{A}_\lambda = \left\{\sup_{t \in [0, T]}\left| W_t - \lambda t \right| \leq 1\right\}, \quad \P(\mathcal{A}_\lambda) > 0.
    }
    \end{equation}
    If $\bracketsig[t][X]{\bsigma_t}$ was a deterministic function of $X_t$, so that $X_t$ was Markovian, we would write on $\mathcal{A}_\lambda$
    \begin{equation}
        X_t \geq \lambda t - 1 + \rho \int_0^t \bracketsig[s][X]{\bsigma}\, ds
    \end{equation}
    and apply the comparison theorem for ODEs to establish the explosion of $X_t$.
    However, in general, it is not true, and the equation becomes path-dependent. Our strategy consists of isolating the term \(\frac{\sigma^{\word{2}\conpow{N}}}{N!}X_t^N\)
    \begin{equation}
        \bracketsig[t][X]{\bsigma} = \dfrac{\sigma^{\word{2}\conpow{N}}}{N!}X_t^N + \bracketsig[t][X]{\bsigma - \sigma^{\word{2}\conpow{N}}\word{2}\conpow{N}}
    \end{equation}
    and bounding the remainder using \eqref{eq:vol_domination_ineq}. Then, we compare the solutions on $[ \tau_n^X, \tau_{n+1}^X )$, bounding $\max\limits_{s \in [0, t]} |X_s|^{N-1}$ by $(n + 1)^{N-1}$, and apply the comparison principle for ODEs on this interval. This leads to a bound on $(\tau_{n+1}^X - \tau_n^X)$, which is used in its turn to bound $\tau_\infty^X$ on $\mathcal{A}_\lambda$.

    Assume first that $\rho\sigma^{\word{2}\conpow{N}} > 0$. 
    By \eqref{eq:vol_domination_ineq}, for $t \in [\tau_n^X, \tau_{n+1}^X)$, we have on $\mathcal{A}_\lambda$
    \begin{align}
        X_{t} &\geq X_{\tau_n^X} + (W_t - W_{\tau_n^X}) + \dfrac{\rho\sigma^{\word{2}\conpow{N}}}{N!}\int_{\tau_n^X}^tX_s^N\,ds - |\rho|C_{\bsigma, T}\int_{\tau_n^X}^t \left(1 + \sup_{u \in [0,\, s]} \left|X_u\right|^{N-1}\right)\,ds \\
        &\geq n + \lambda (t - \tau_{n}^X) - 2 + \dfrac{\rho\sigma^{\word{2}\conpow{N}}}{N!}\int_{\tau_n^X}^t X_s^N \,ds - |\rho|C_{\bsigma, T} \left(1 + (n+1)^{N-1}\right)(t - \tau_n^X).
    \end{align}
    This allows us to apply the comparison theorem for ODEs, so that $X_t \geq Y_n(t)$ on $[\tau_n^X, \tau_{n+1}^X)$, where $Y_n(t)$ satisfies
    {
    \begin{equation}
        \dot Y_n(t) = F_n(Y_n(t)), \quad Y_n(\tau_n^X) = n - 2,
    \end{equation}
    where
    $$
    F_n(x) := \lambda - |\rho|C_{\bsigma, T} \left(1 +  (n+1)^{N-1}\right) + \dfrac{\rho\sigma^{\word{2}\conpow{N}}}{N!} x^N 
    $$
    At the left endpoints of the intervals $[\tau_n^X, \tau_{n+1}^X)$, the derivatives  
    $$
    \dot Y_n(\tau_n^X) = F_n(n-2) := \lambda - |\rho|C_{\bsigma, T} \left(1 +  (n+1)^{N-1}\right) + \dfrac{\rho\sigma^{\word{2}\conpow{N}}}{N!} (n-2)^N
    $$
    are uniformly bounded from below in \( n \), since \( \rho\sigma^{\word{2}\conpow{N}} > 0 \).
    Hence, one can choose \(\lambda\) sufficiently large so that \(\dot Y_n(\tau_n^X) \ge 0\) for all \(n \ge 0\).
    
    Moreover, as each \(F_n(\cdot)\) is increasing and \(F_n(Y_n(\tau_n^X)) > 0\), it follows that \(F_n(Y_n(t)) > 0\) for all \(t \ge \tau_n^X\).
    Consequently, the solutions \((Y_n(t))_{t \ge \tau_n^X}\) remain increasing and satisfy \(Y_n(t) \ge n - 2\) for all \(n \in \N\) and \(t \in [\tau_n^X, \tau_{n+1}^X)\).
    Therefore,}
    \begin{equation}
        Y_n(t) \geq n - 2 + \left(\lambda + \dfrac{\rho\sigma^{\word{2}\conpow{N}}}{N!}(n-2)_+^N - |\rho|C_{\bsigma, T} \left(1 + (n+1)^{N-1}\right)\right)(t - \tau_n^X).
    \end{equation}
    Hence, on $\mathcal{A}_\lambda$, the following estimate holds
    \begin{equation}
        \tau_{n+1}^X - \tau_{n}^X \leq \tau_{n+1}^Y - \tau_{n}^X \leq \dfrac{3}{\lambda + \frac{\rho\sigma^{\word{2}\conpow{N}}}{N!}(n-2)_+^N - |\rho|C_{\bsigma, T} \left(1 + (n+1)^{N-1}\right)}
    \end{equation}
    and
    \begin{equation}
        \tau_{\infty}^X = \sum_{n \geq 0}(\tau_{n+1}^X - \tau_{n}^X) \leq
         \sum_{n \geq 0} \dfrac{3}{\lambda + \frac{\rho\sigma^{\word{2}\conpow{N}}}{N!}(n-2)_+^N - |\rho|C_{\bsigma, T} \left(1 + (n+1)^{N-1}\right)}.
    \end{equation}
    The series converges since $N \geq 2$ and the its sum can be made arbitrary small by choosing $\lambda$ large enough.

    For the remaining case of even $N$ and $\rho\sigma^{\word{2}\conpow{N}} < 0$, we take $\lambda < 0$ large enough, so that $X_{\tau_{n}^X} = -n$ on $\mathcal{A}_\lambda$, and we have on $[\tau_{n}^X, \tau_{n+1}^X)$
    \begin{align}
        X_{t} &\leq X_{\tau_n^X} + (W_t - W_{\tau_n^X}) + \dfrac{\rho\sigma^{\word{2}\conpow{N}}}{N!}\int_{\tau_n^X}^tX_s^N\,ds + |\rho|C_{\bsigma, T}\int_{\tau_n^X}^t \left(1 + \sup_{u \in [0,\, s]} \left|X_u\right|^{N-1}\right)\,ds \\
        &\leq -n + \lambda (t - \tau_{n}^X) + 2 + \dfrac{\rho\sigma^{\word{2}\conpow{N}}}{N!}\int_{\tau_n^X}^t X_s^N \,ds + |\rho|C_{\bsigma, T} \left(1 + (n+1)^{N-1}\right)(t - \tau_n^X).
    \end{align}
    As $\dfrac{\rho\sigma^{\word{2}\conpow{N}}}{N!}X_s^N < 0$, using the same argument as before, we can choose $\lambda(T)$ in order to ensure that $X_t$ blows up to $-\infty$ on $\mathcal{A}_\lambda$ before $T$ for any $T > 0$.

\subsection{Proof of necessity in Theorem~\ref{T:main_martingality_multid} \textit{(ii)}}\label{section:proof_local_mart_mulitiD}

We prove the result for the case where $ \sigma^{\word{2}\conpow{N}} > 0$ showing that $ X $ blows up to $ +\infty $ with positive probability. The case where $ \sigma^{\word{2}\conpow{N}} < 0 $ and $ N $ is even is treated similarly by showing that $ X $ blows up to $ -\infty $.

{The idea of this proof is to extend the proof of strict local martingality in Section~\ref{section:proof_local_martingality} by finding a suitable event $\mathcal{A}$ of positive probability such that the solution of \eqref{eq:sdeX} blows up before $T > 0$ on $\mathcal{A}$. We achieve this by keeping the elements of the signature containing $Z$ (i.e., those corresponding to words containing the letters $\word{3}, \ldots, \word{d + 2}$) small, so that the solution of the SDE remains close to the one considered in the one-dimensional case. More precisely, we define  
$$
\tau_n^X = \inf\{t \geq t_0\colon X_t = n\},
$$
and we fix $t_0 > 0$, $M \in \mathbb{N}$, and $\epsilon$ such that $N\epsilon < 1$. We require that on $\mathcal{A}$, the following inequalities hold:
\begin{align}
X_{t_0} &\geq M, \label{eq:cond1}\\
| \langle \word{w},\widehat{\mathbb{Y}}^X_{t_0} \rangle | &\leq 1, \quad \text{for all $\word{w}$ containing at least one of the letters $\word{3}, \ldots, \word{d + 2}$;} \label{eq:cond2}\\
\sup_{s \in [\tau_n^X, \tau_{n+1}^X{ \land\tilde\tau_{n+1}^X]}} | \langle \word{w} , \widehat{\mathbb{Y}}^X_{s} \rangle | &\leq n^{N_{\word{2}}(\word{w})+\epsilon |\word{w}|}, \quad \text{for $n \geq M$ and for all words $\word{w}$ such that $|\word{w}| \leq N$;} \label{eq:cond3}\\
\sup_{s \in [\tau_n^X, \tau_{n+1}^X{ \land\tilde\tau_{n+1}^X]}} |B_s - B_{\tau_n}| &\leq n^{-N/2+\epsilon}, \quad \text{for $n \geq M$,} \label{eq:cond4}
\end{align}
{ where $\tilde\tau_{n+1}^X := \tau_n^X + Cn^{-N}$ for some positive constant $C = C_{\bsigma, N} > 0$ to be determined later.}
\paragraph{Step 1: Proving the explosion on $\mathcal{A}$.} 
Conditional on this event, we can use ODE comparison on the dynamics of $X$ { to choose $C > 0$ such that}
\begin{equation}\label{eq:d_tau_bound}
    \tau_{n+1}^X-\tau_n^X \leq C n^{-N}, \quad n \geq M,
\end{equation}
{ so that $\tau_{n+1}^X\land\tilde\tau_{n+1}^X = \tau_{n+1}^X$ on $\mathcal{A}$ for $n \geq M$.}
Indeed, for $t \in [\tau_n^X, \tau_n^X + C n^{-N}]$, we have 
\begin{align*}
X_t &= X_{\tau_n^X} + \int_{\tau_n^X}^t \bracketsigX[s][\widehat{Y}^X]{\bsigma}\,ds + B_t - B_{\tau_n^X}\\ 
&\geq X_{\tau_n^X} + \frac{\sigma^{\word{2}^{\otimes N}}}{N!} \int_{\tau_n^X}^t X_s^N ds - C_{\bsigma} n^{(N - 1) + \epsilon N} (t-\tau_n^X) - n^{-N/2 + \epsilon} \\
&\geq n - \left(C_{\bsigma} C n^{N\epsilon - 1} + n^{-N/2 + \epsilon}\right) + \frac{\sigma^{\word{2}^{\otimes N}}}{N !} \int_{\tau_n^X}^t X_s^N ds,
\end{align*}
where we used conditions \eqref{eq:cond3} and \eqref{eq:cond4} in the first inequality. For $n$ large enough, we have 
$$
C_{\bsigma} C n^{N\epsilon - 1} + n^{-N/2 + \epsilon} < 1,
$$
since $N\epsilon < 1$, so that
$$
X_t \geq (n-1) + \frac{\sigma^{\word{2}^{\otimes N}}}{N !} \int_{\tau_n^X}^t X_s^N ds.
$$
By the comparison theorem for ODEs, this implies that $\tau_{n+1}^X - \tau_n^X\leq {\tau}_{(n-1) \to (n+1)}^x$, where ${\tau}_{(n-1) \to (n+1)}^x$ is the time needed for the ODE 
$$\dot{x}(t) =\frac{\sigma^{\word{2}^{\otimes N}}}{N !} x^N(t),$$ 
to go from $(n-1)$ to $(n+1)$. Since the coefficient $\sigma^{\word{2}^{\otimes N}}$ is positive, a straightforward computation gives
\begin{equation} \label{eq:tildetauandc}
{\tau}_{(n-1) \to (n+1)}^x =  \left( \frac{N !}{(N-1) \sigma^{\word{2}^{\otimes N}}}\right) \left((n-1)^{1-N} - (n+1)^{1-N} \right)\leq C n^{-N}
\end{equation}
for a suitable choice of $C > 0$. It then follows that, on the event $\mathcal{A}$, we have 
$$
\tau_\infty \leq t_0 + C \sum\limits_{n\geq M} n^{-N}.
$$ 
Since $t_0$ and $M$ in \eqref{eq:cond1}--\eqref{eq:cond2} are arbitrary, they can be chosen so that $\tau_\infty < T$ on $\mathcal{A}$.
It remains to show that $\P(\mathcal{A}) > 0$.
\paragraph{Step 2: Showing that $\P(\mathcal{A}) > 0$.} 
First, we note that 
\begin{equation*}
    \mathcal{A} = A_M \cap \bigcap_{n\geq M}\left(A_n^B \cap \bigcap_{\word{w}:\, |\word{w}| \leq N}A_n^{\word{w}}\right), 
\end{equation*}
where the event $A_M$ corresponds to the inequalities \eqref{eq:cond1}--\eqref{eq:cond2}, $A_n^{\word{w}}$ corresponds to \eqref{eq:cond3} for fixed $n$ and $\word{w}$, and $A_n^B$ corresponds to \eqref{eq:cond4} for a given $n$.

{We first claim that $\P\left(A_M\right) > 0$. By the support theorem of \citet[Theorem 5.1]{SV72}, it suffices to show that, for a suitable choice of $(w,z^1,\ldots, z^d)$ (deterministic smooth functions on $[0,T]$), the corresponding deterministic dynamics (i.e. where $W$, $Z^i$ are replaced by $w, z^i$ in \eqref{eq:sdeX}) satisfy the conditions. Taking $z^1 \equiv \cdots \equiv z^d \equiv 0$ clearly leads to \eqref{eq:cond2}. Further taking $w_t = \lambda t$ leads to \eqref{eq:cond1} by the analysis of Section \ref{section:proof_local_mart_1D}. It follows that   $\P\left(A_M\right) > 0$.
}

We define 
$$
A_{n + 1} := A_n \cap \left(\bigcap_{|\word{w}| \leq N} A^{\word{w}}_{n + 1}\right) \cap A^B_{n + 1}, \quad n \geq M.
$$
It is enough to prove that, if $M$ is large enough, for each $n \geq M$, we have
\begin{equation} \label{eq:cond}
    \P(A_{n+1}^c|A_n) \leq c e^{-c n^{\delta}},
\end{equation}
for some $c,\delta>0$. Indeed, since $A_{n+1} \subset A_n$ and $\mathcal{A} = \bigcap\limits_{n \geq M} A_n$, we write
\begin{equation}
    \P(\mathcal{A}) = \P(A_M)\prod_{n \geq M}\P(A_{n+1} | A_n),
\end{equation}
so that, if \eqref{eq:cond} is verified, we have
$$
    \log(\P(\mathcal{A})) = \log(\P(A_M)) + \sum_{n \geq M}\log(\P(A_{n+1} | A_n)) \geq  \log(\P(A_M)) + \sum_{n \geq M}\log(1 - c e^{-c n^{\delta}}).
$$
As $n \to \infty$, $\log(1 - c e^{-c n^{\delta}}) \sim c e^{-c n^{\delta}}$, and the series converges. Thus, we obtain $\P(\mathcal{A}) > 0$.

Note that since
\begin{equation*}
    \P(A_{n+1}^c|A_n) = \P\left(\bigcup_{\word{w}:\, |\word{w}| \leq N} (A_{n+1}^{\word{w}})^c \cup (A_{n+1}^B)^c  \Bigg|\, A_n\right)
    \leq \P((A_{n+1}^B)^c | A_n) + \sum_{\word{w}:\, |\word{w}| \leq N} \P((A_{n+1}^{\word{w}})^c |A_n),
\end{equation*}
the condition \eqref{eq:cond} is stable under finite intersections, and it suffices to prove \eqref{eq:cond} for $\P((A_{n+1}^B)^c|A_n)$ and for each of the $\P((A_{n+1}^{\word{w}})^c|A_n)$.

The bound \eqref{eq:cond} for $\P(A_{n+1}^B|A_n)$ follows from Gaussian scaling and Gaussian tail estimates for the Brownian motion $B$.
Thus, it remains to establish \eqref{eq:cond} for the events $A_{n+1}^{\word{w}}$, which we proceed to do via induction on the length of the word $|\word{w}|$. 

For $|\word{w}| = 1$, we consider three cases. If $\word{w} = \word{1}$, \eqref{eq:cond3} holds for sufficiently large $n$ due to the bound \eqref{eq:d_tau_bound} on $\tau_n^X$. If $\word{w} = \word{2}$, \eqref{eq:cond3} follows directly from the definition of $\tau_n^X$ for large $n$. In the remaining case, where $\word{w} \in \{\word{3}, \ldots, \word{d + 2}\}$, \eqref{eq:cond} follows from Gaussian tail estimates, as was previously done for the probability $\P(A_{n+1}^B|A_n)$. 

Now, suppose that $|\word{w}| > 1$ and that \eqref{eq:cond} holds for all words $\word{u}$ such that $|\word{u}| < |\word{w}|$.

By Proposition~\ref{prop:nod}, we can write $\word{w}$ as a linear combination of shuffle products and we consider three cases.

The first case consists of words of the form $\word{w} = \word{w_{1}} \shuprod \ldots \shuprod \word{w_k}$ with $k \geq 2$. The result then follows from the induction hypothesis for $\word{w_j}$, since on $\bigcap\limits_{j = 1}^kA_{n+1}^{\word{w_j}}$ we have
\begin{equation*}
    \sup_{s \in [\tau_n^X, \tau_{n+1}^X{ \land\tilde \tau_{n+1}^X}]} | \langle \word{w} , \widehat{\mathbb{Y}}^X_{s} \rangle| \leq \prod_{j=1}^k \sup_{s \in [\tau_n^X, \tau_{n+1}^X { \land\tilde \tau_{n+1}^X}]} | \langle \word{w_j} , \widehat{\mathbb{Y}}^X_{s} \rangle| \leq \prod_{j=1}^k n^{N_{\word{2}}(\word{w_j})+\epsilon |\word{w_j}|} = n^{N_{\word{2}}(\word{w})+\epsilon |\word{w}|}.
\end{equation*}
This implies that $A_{n+1}^{\word{w}} \subset \bigcap\limits_{j = 1}^kA_{n+1}^{\word{w_j}}$. 

The second case concerns words of the form $\word{u} \word{1}$, which is again handled by the induction hypothesis for $\word{u}$. Indeed, conditional on $A_n$ and $A_{n+1}^{\word{u}}$, we have
\begin{align*}
    \sup_{s \in [\tau_n^X, \tau_{n+1}^X{ \land\tilde \tau_{n+1}^X}]} | \langle \word{w} , \widehat{\mathbb{Y}}^X_{s} \rangle| &= \sup_{s \in [\tau_n^X, \tau_{n+1}^X{ \land\tilde \tau_{n+1}^X}]} \left| \langle \word{w} , \widehat{\mathbb{Y}}^X_{\tau_n^X}\rangle + \int_{\tau_n^X}^s\langle \word{u} , \widehat{\mathbb{Y}}^X_{s} \rangle ds\right| \\
    &\leq (n - 1)^{N_{\word{2}}(\word{w})+\epsilon |\word{w}|} + ({ \tilde\tau_{n+1}^X} - \tau_n^X)n^{N_{\word{2}}(\word{u})+\epsilon |\word{u}|} \\
    &\leq (n - 1)^{N_{\word{2}}(\word{w})+\epsilon |\word{w}|} + Cn^{-N + N_{\word{2}}(\word{u})+\epsilon |\word{u}|} \\
    &\leq (n - 1)^{N_{\word{2}}(\word{w})+\epsilon |\word{w}|} + Cn^{-1+\epsilon (N - 1)},
\end{align*}
which is bounded by $n^{N_{\word{2}}(\word{w})+\epsilon |\word{w}|}$ for sufficiently large $n$ (recall that $M$ can be chosen large enough).

The last case corresponds to $\word{w} = \word{u} \word{i}$ for $i \in\{3, \ldots, d + 2\}$. 
We write
\begin{align*}
    \P((A_{n+1}^{\word{w}})^c|A_n) &= \P((A_{n+1}^{\word{w}})^c, (A_{n+1}^{\word{u}})^c|A_n) + \P((A_{n+1}^{\word{w}})^c, A_{n+1}^{\word{u}}|A_n) \\
    &\leq \P((A_{n+1}^{\word{u}})^c|A_n) + \P((A_{n+1}^{\word{w}})^c, A_{n+1}^{\word{u}}|A_n),
\end{align*}
and note that \eqref{eq:cond} holds for the first term by the induction hypothesis for $\word{u}$.
For $s \geq \tau_n$, we have
\[
\langle \word{w} , \widehat{\mathbb{Y}}^X_{s} \rangle = \langle \word{w} , \widehat{\mathbb{Y}}^X_{\tau_n^X} \rangle  + \int_{\tau_n}^s \langle \word{u} , \widehat{\mathbb{Y}}^X_{s}\rangle {\circ} dZ^{i - 2}_s.
\]
The first term and the Stratonovich correction are bounded on $A_n$ and $A_{n+1}^{\word{u}}$ as in the previous case. Thus, it suffices to prove the bound for the conditional probability of the event
\[
\left\{\sup_{t \in [\tau_n^X, \tau_{n+1}^X{ \land \tilde \tau_{n+1}^X}]}  \int_{\tau_n^X}^{t} \langle \word{u} , \widehat{\mathbb{Y}}^X_{s}\rangle dZ^{i - 2}_s \geq n^{N_{\word{2}}(\word{w}) + \epsilon |\word{w}|-1}\right\}.
\]
{
We use the exponential inequality
\begin{equation}\label{eq:exp_ineq}
    \P\left( \sup_{t \in [\tau_1, \tau_2]}\left|\int_{\tau_1}^{t} \gamma_s dW_s \right|\geq \lambda, \; \sup_{t \in [\tau_1, \tau_2]} |\gamma_t| \leq \eta\right) \leq \exp\left(- \frac{\lambda^2}{2 \eta^2 h}\right),
\end{equation}
which holds for all stopping times \(\tau_1 \leq \tau_2\leq \tau_1 + h \), any \(h > 0\), and any progressively measurable \(\gamma\). Indeed, it follows from 
\begin{align*}
    &\P\left( \sup_{t \in [\tau_1, \tau_2]}\left|\int_{\tau_1}^{t} \gamma_s dW_s \right|\geq \lambda, \; \sup_{t \in [\tau_1, \tau_2]} |\gamma_t| \leq \eta\right) \\
    &\leq \P\left( \sup_{t \in [\tau_1, \tau_2]}\left|\int_{\tau_1}^{t} \gamma_s dW_s \right|\geq \lambda, \; \int_{\tau_1}^{\tau_2} |\gamma_s|^2\,ds \leq \eta^2h\right)
    \\& = \P\left( \sup_{t \geq 0}M_t \geq \lambda, \; \langle M \rangle_{\infty} \leq \eta^2h \right) \leq \exp\left(- \frac{\lambda^2}{2 \eta^2 h}\right),
\end{align*}
where we applied the exponential Bernstein inequality (cf.~\citet[IV, 3.16]{RevuzYor1999}) to the martingale
$M_t := \int_{\tau_1}^{(\tau_1 + t)\land\tau_2} \gamma_s \, dW_s,\  t \geq 0.
$
Applying the exponential inequality \eqref{eq:exp_ineq} with $h = Cn^{-N}$, we then obtain
}
\begin{align*}
    \P &\left(\sup_{t \in [\tau_n^X, \tau_{n+1}^X{ \land \tilde \tau_{n+1}^X}]}  \int_{\tau_n}^{t} \langle \word{u} , \widehat{\mathbb{Y}}^X_{s}\rangle dZ^i_s \geq n^{N_{\word{2}}(\word{w}) + \epsilon |\word{w}|-1}, A^{\word{u}}_{n+1} \Big| A_n\right) \\
         &\leq \P\left( \sup_{t \in [\tau_n^X,  \tau_{n}^X{ \land\tilde \tau_{n+1}^X}]}  \int_{\tau_n}^{t} \langle \word{u}, \widehat{\mathbb{Y}}^X_{s}\rangle {\circ} 
     dZ^i_s \geq n^{N_{\word{2}}(\word{w}) + \epsilon |\word{w}|-1}, \; \sup_{t \in [\tau_n^X, \tau_{n}^X{ \land\tilde \tau_{n+1}^X]}} \left| \langle \word{u} , \widehat{\mathbb{Y}}^X_{t}\rangle \right| \leq n^{N_{\word{2}}(\word{w}) + \epsilon (|\word{w}|-1)} \Big| A_n \right) \\
     &\leq \exp\left(- \frac{n^{N-2(1-\epsilon)}}{2{ C}}\right).
\end{align*}
Thus, \eqref{eq:cond} holds for $\word{w}$. This completes the proof.
}

\section{Proof of Theorem \ref{thm:moments}: Moments explosion}\label{sect:moments}

Without loss of generality, we will assume that $N$ is odd and
\[
\sigma^{\word{2}\conpow{N}}  > 0 \mbox{ and }\rho < 0.
\]
The proof for the remaining case $ \sigma^{\word{2}\conpow{N}} < 0 $ and $ \rho > 0 $ is analogous.  

By It\^o's formula,
\begin{align}
    S_T^m &= S_0^m \exp\left( m\int_0^T\sigma_s\,dB_s - \dfrac{m}{2}\int_0^T\sigma_s^2\,ds \right) \\ &= S_0^m \exp\left( \rho m\int_0^T\sigma_s\,dW_s + \bar\rho m\int_0^T\sigma_s\,dW_s^\perp - \dfrac{m}{2}\int_0^T\sigma_s^2\,ds \right),
\end{align}
As $\sigma_t = \bracketsig{\bsigma}$ is adapted to the filtration generated by $W$, conditioning with respect to $W$ implies
\begin{align*}
S_0^{-m} \E\left[S_T^m\right] &=   \;\E \left[ \exp\left(\rho m \int_0^T \bracketsig[s]{\bsigma} \,d W_s + \frac{\bar{\rho}^2 m^2 - m}{2} \int_0^T \bracketsig[s]{\bsigma}^2 \, ds\right)\right].
\end{align*}
{We then proceed as in \cite{Gassiat2018OnTM} and apply the \cite{BD98} formula {(more precisely, a variant allowing unbounded functionals, see Appendix \ref{sec:BD})}
to obtain}
\begin{equation}\label{eq:BD}
\begin{aligned} 
 \ln {\E}\left[S_T^m/S^m_0\right]  &= \sup_{(u_t)_{t\geq0} \in \mathcal{U}_q} {\E} \left[ \int_0^T  \left( - \frac{u_s^2}{2} + \rho m  \left\langle \bsigma, \widehat{\mathbb{X}}^u_s \right\rangle u_s +   \frac{\bar{\rho}^2 m^2 - m}{2} \left\langle \bsigma, \widehat{\mathbb{X}}^u_s \right\rangle^2 \right)ds\right] \\
 & = \sup_{(u_t)_{t\geq0} \in \mathcal{U}_q} {\E} \left[ \int_0^T   \left( \frac{ m^2 - m}{2} \left\langle \bsigma, \widehat{\mathbb{X}}^u_s \right\rangle^2 - \frac{\left(u_s - \rho m \left\langle \bsigma, \widehat{\mathbb{X}}^u_s \right\rangle\right)^2}{2} \right)ds\right]  \\
 & =: \sup_{u \in \mathcal{U}_q}\mathcal{J}(u;\rho,m,T) \; =: \mathcal{V}(\rho,m,T),
\end{aligned}
\end{equation}
where 
\begin{align}\label{eq:Uq}
\mathcal{U}_{q}=\left\{ (u_t)_{t \geq 0} \mbox{ progressively measurable with }{\E}\left[ \int_0^T |u_t|^{q} dt \right] < +\infty \right\},
    \end{align}
for some fixed  $q \geq 2N$, 
and $\widehat{\mathbb{X}}^u_t$ denotes the signature of $\widehat{X}^u_t = (t, X^u_t)$ with $X^u_t$ defined by
\[
X^u_t = W_t + \int_0^{t} u_s ds.
\]

The finiteness of the moments is therefore equivalent to the control problem having a finite value. We now proceed to prove that, depending on the position of $ \rho $ relative to $ -\sqrt{1 - \frac{1}{m}} $, the value will either be finite or infinite.

\subsection{Proof of Theorem \ref{thm:moments} (i): negative result}

In this subsection we assume that
\[
- \sqrt{1-\frac{1}{m}} < \rho < 0,
\]
and we will show that $\mathcal{V}(\rho,m,T) = + \infty$. 
The idea is simple: we take $u$ of feedback form
\[u_t = \lambda \left\langle \bsigma, \widehat{\mathbb{X}}^u_t \right\rangle \]
for $\lambda \in \R$.

In that case, we denote by $X^\lambda$ the solution of the { signature SDE}
\[
dX^\lambda_t = dW_t +  \lambda \left\langle \bsigma, \widehat{\mathbb{X}}^\lambda_t \right\rangle dt,
\]
and we denote its signature by $\widehat{\mathbb{X}}^\lambda_t$. This solution exists (up to a possible explosion time) by Proposition \ref{prop:sigODE},
the value \eqref{eq:BD} corresponding to the control is simply
\[
 P (\lambda) \E \left[ \int_0^T \left\langle \bsigma, \widehat{\mathbb{X}}^\lambda_s \right\rangle^2 ds \right], 
\]
with
\[ P(\lambda) = \frac{m^2 - m}{2} - \frac{(\lambda - \rho m)^2}{2}.
\]
The assumption on $\rho$ is exactly equivalent to $P(0)>0$, which guarantees that there exists $\lambda$ s.t. 
\begin{equation} \label{eq:lambda}
\lambda > 0 \mbox{ and } P(\lambda) >0. 
\end{equation}
As shown in the proof in Section \ref{section:proof_local_martingality}, the first inequality $\lambda > 0$ will imply that $X^\lambda$ blows up with positive probability, which with the help of the second one implies that the problem value \eqref{eq:BD} is infinite.

The above argument is not completely rigorous, since the feedback control is only defined up to the explosion time and clearly not in $\mathcal{U}_q$. This can be dealt with by a simple truncation argument : for $R>0$, let 
{\[
\tau_{R,\lambda} = \inf\{ t \geq 0 :  |X_t^\lambda|= R \}
\]
and take 
\[
u^{\lambda,R}_t =  \begin{cases} \lambda \left\langle \bsigma, \widehat{\mathbb{X}}^\lambda_t \right\rangle, & t < \tau_{R,\lambda} \\
0, & t \geq \tau_{R,\lambda} \end{cases} 
\]

Then, by Lemma \ref{lem:expYV}, $u^{\lambda,R}$ is in $\mathcal{U}_q$ (for any $q \geq 2$), and 
\begin{align*}
   \mathcal{V}(\rho,m,T) \geq \mathcal{J}(u^{\lambda,R};\rho,m,T) &= P(\lambda) \E\left[\int_0^{T \wedge \tau_{R,\lambda}} \left\langle \bsigma, \widehat{\mathbb{X}}^\lambda_s \right\rangle^2 ds \right] + P(0) \E \left[ \int^T_{T \wedge \tau_{R,\lambda}} \left\langle \bsigma, \widehat{\mathbb{X}}_s^{u^{\lambda,R}} \right\rangle^2 ds\right]. 
\end{align*}
Taking the supremum over $R$, we obtain
\[\mathcal{V}(\rho,m,T) \geq P(\lambda) { \E}\left[\int_0^{T \wedge \tau_{\infty,\lambda}} \left\langle \bsigma, \widehat{\mathbb{X}}^\lambda_s \right\rangle^2 ds \right] \]
where $\tau_{\infty,\lambda}$ is the explosion time of $X^{\lambda}$. However, the r.h.s is infinite since the event $\{\tau_{\infty,\lambda} < T \}$ has positive probability (by the results of Section~\ref{section:proof_local_mart_1D}), and on this event $\int_0^{\tau_{\infty,\lambda}}  \left\langle \bsigma, \widehat{\mathbb{X}}^\lambda_s \right\rangle^2 ds = +\infty$ (by Proposition \ref{prop:sigODE}).
}
\subsection{Proof of Theorem \ref{thm:moments} (ii): positive result}

We assume now that the condition
\begin{equation}\label{eq:corr_condition_pos}
    \rho < - \sqrt{1-\frac{1}{m}},
\end{equation}
is verified and fix $(u_t)_{t\geq0} \in \mathcal{U}_q$. We want to bound
\begin{align}\label{eq:value_to_bound}
{\E} \left[ \int_0^T  \left( - \frac{u_s^2}{2} + \rho m \left\langle \bsigma, \widehat{\mathbb{X}}^u_s \right\rangle u_s +   \frac{\bar{\rho}^2 m^2 - m}{2} \left\langle \bsigma, \widehat{\mathbb{X}}^u_s \right\rangle^2 \right)ds\right].
\end{align}
Define 
\begin{equation}
    U_t := \int_0^{t} u_s\, ds.
\end{equation}
We denote by $\widehat Y^U_t=(t, U_t, W_t)$ and we denote its signature by $\widehat{\mathbb{Y}}^U_t$. Since this is a $3$-dimensional process, this signature is in the dual of the tensor algebra over words with letters $\word{1},\word{2},\word{3}$. Since $X = U + W$, we can rewrite the linear forms in terms of $\widehat{\mathbb{Y}}^U_t$. Namely, there exists $\tilde{\bsigma} \in \tTA[3]{N}$, such that $\left\langle \bsigma, \widehat{\mathbb{X}}^u_t \right\rangle = \left\langle \tilde\bsigma,\widehat{\mathbb{Y}}^U_t \right\rangle$.

In the expectation \eqref{eq:value_to_bound}, the term with the highest degree in $U$ is  
\[
\frac{\bar{\rho}^2 m^2 - m}{2} \left(\frac{\bsigma^{\word{2}\conpow{N}}}{N!}\right)^2 \mathbb{E} \left[\int_0^T U_t^{2N} dt\right].
\]  
Note that ${\bar{\rho}^2 m^2 - m}$ is strictly negative by assumption.  

We return to the value in \eqref{eq:value_to_bound} and bound the three terms in the integrand separately. The first one is non-positive and can therefore be ignored.

\paragraph{Step 1. Bounding the third term.} Using the same techniques as in the proof of Theorem~\ref{T:main_martingality_multid}, we can bound the third term $\frac{\bar{\rho}^2 m^2 - m}{2} \E\left[\int_0^T\left\langle\tilde\bsigma, \widehat{\mathbb{Y}}^U_s \right\rangle^2ds\right]$ from above by
\begin{equation}\label{eq:third_term_bound}
    - c_1 \E\left[\int_0^T  U_t^{2N} dt\right] + \sum_{\substack{|\word{w}| \leq N,\\ N_{\word{2}}(\word{w})<N}} c_{\word{w}} \E\left[\int_0^T\left\langle\word{w}, \widehat{\mathbb{Y}}^U_t\right\rangle^{2}\, dt\right]
\end{equation}
for some positive $c_1$ and $c_{\word{w}}$. By Lemma \ref{lem:expYV}, it is further bounded from above by
\[
\E \left[\int_0^T \left(- c_1 U_t^{2N} + c_2 U_t^{2N-1}\right)dt\right] + c_3 < +\infty.
\]

\paragraph{Step 2. Bounding the second term.} For the second term $\rho m{\E} \left[ \int_0^T \left\langle \tilde\bsigma, \widehat{\mathbb{Y}}^U_s \right\rangle u_s ds\right]$, we note that, for suitable coefficients $c_{\word{w}}$,
\begin{equation}\label{eq:second_term_bound}
    \rho m{\E} \left[ \int_0^T \left\langle \bsigma, \widehat{\mathbb{X}}^u_s \right\rangle u_s ds\right] = \rho m{\E} \left[ \int_0^T \left\langle \bsigma, \widehat{\mathbb{X}}^u_s \right\rangle dU_s\right] = -c_4 \E[U_T^{N+1}] + \sum_{\substack{|\word{w}| \leq N,\\ N_{\word{2}}(\word{w})<N}} c_{\word{w}} \E\left[ \left\langle  \word{w} \word{2},\widehat{\mathbb{Y}}^U_T\right\rangle \right],
\end{equation}
where $c_4 = -\dfrac{\rho m\bsigma^{\word{2}\conpow{N}}}{(N + 1)!}$.
We note that the Itô and Stratonovich integrals with respect to $ U $ coincide since $ U $ has bounded variation.  
By the assumption of the theorem, $c_4 > 0$, so that the first term in \eqref{eq:second_term_bound} is negative.

For the second term, we claim that it holds, for any $\word{w}$ with $|\word{w}| \leq N+1$ and $N_{\word{2}}(\word{w}) \leq N$, for arbitrary $\epsilon>0$,
\begin{equation} \label{eq:term2}
    \E\left[ \left\langle  \word{w},\widehat{\mathbb{Y}}^U_T\right\rangle \right]  \lesssim 1+ \E[ |U_T|^{N+\epsilon}] + \int_0^T \E[|U_t|^{2(N-1)+\epsilon}] \, dt
\end{equation}
We first consider separately the case where $|\word{w}| = N + 1$ and $N_{\word{2}}(\word{w}) = N $, i.e., only one letter differs from $ \word{2} $, say $ \word{j} $. By Proposition~\ref{prop:nod}, we can express $ \word{w} $ in the form
\begin{equation}\label{eq:w_decomp_1}
    \word{w} = \sum_{k=0}^N \alpha_k \word{2}\shupow{k}\shuprod \word{2}\conpow{N - k}\word{j}
\end{equation}

If $\word{j} = \word{3}$, then \eqref{eq:w_decomp_1} takes the form
\[
\E\left[ \left\langle  \word{w},\widehat{\mathbb{Y}}^U_T\right\rangle \right] =\sum_{k=1}^N \alpha_k \E \left[U_T^k \left(\int_0^T U_t^{N-k} dW_t\right)\right]
\]
Note that we omitted the term for $ k = 0 $, as it has zero expectation.  
Applying Young's inequality and the Burkholder--Davis--Gundy inequality to each term, we obtain   
{\begin{align}
    \E \left[U_T^k \left(\int_0^T U_t^{N-k} dW_t\right)\right] &\lesssim \E\left[ |U_T|^{N+\epsilon}\right] + \E \left[\left(\int_0^T U_t^{(N-k)}dW_t\right)^{\frac{N+\epsilon}{(N -k + \epsilon) }}\right]
    \\ 
    &\lesssim \E\left[ |U_T|^{N+\epsilon}\right] + \E \left[\left(\int_0^T U_t^{2(N-k)}dt\right)^{\frac{N+\epsilon}{2(N -k + \epsilon) }}\right]
    \\ 
    &\lesssim 1 + \E\left[ |U_T|^{N+\epsilon}\right] + \E \left[\left(\int_0^T U_t^{2(N-k)}dt\right)^{\max\left\{\frac{N+\epsilon}{2(N -k + \epsilon) }, 1\right\}}\right]
    \\
    &\lesssim 1 + \E \left[|U_T|^{N+\epsilon}\right] + \E \left[\int_0^T |U_t|^{\max\left\{ N+\epsilon, 2(N-k)\right\}} dt\right],
    \\
    &\lesssim 1 + \E \left[|U_T|^{N+\epsilon}\right] + \E \left[\int_0^T |U_t|^{2(N - 1) + \epsilon
} dt\right],
\end{align}
}
which proves the claim \eqref{eq:term2} in this case.

If $\word{j} = \word{1}$, a similar but simpler computation yields
\[
\E \left[U_T^k \left(\int_0^T U_t^{N-k} dt\right)\right] \lesssim \E\left[ |U_T|^{N}\right] +  \E\left[ \int_0^T |U_t|^{N} dt\right],
\]
We then assume that $N_{\word{2}}(\word{w}) \leq N-1$. In this case, by Proposition~\ref{prop:nod}, we can write $\word{w}$ as shuffle polynomial with monomials of the form
\begin{equation}\label{eq:w_decomp_2}
    \word{2}\shupow{k}\shuprod \word{u_11} \shuprod \ldots  \shuprod \word{u_m1} \shuprod \word{v_13}  \shuprod \ldots  \shuprod \word{v_n3},
\end{equation}
so that it remains to prove estimates on
\[
\E\left[U_T^{k} \prod_{i=1}^{m} \left(\int_0^T  \left\langle  \word{u_i},\widehat{\mathbb{Y}}^U_t\right\rangle dt\right)\prod_{j=1}^{n} \left(\int_0^T  \left\langle  \word{v_j},\widehat{\mathbb{Y}}^U_t\right\rangle \circ dW_t\right)\right].
\]
Again, an application of Young's inequality yields that this expectation is bounded by
 \[
 \E[|U_T|^{N_{\word{2}}(\word{w})+\epsilon}] + \sum_{i=1}^m \E\left[\left|\int_0^T \left\langle  \word{u_i},\widehat{\mathbb{Y}}^U_t\right\rangle  dt\right|^{\frac{N_{\word{2}}(\word{w})+\epsilon}{N_{\word{2}}(\word{u_i})+\epsilon_{1,i}}}\right] + \sum_{j=1}^{n} \E\left[\left|\int_0^T \left\langle  \word{v_j},\widehat{\mathbb{Y}}^U_t \right\rangle  \circ dW_t\right|^{\frac{N_{\word{2}}(\word{w})+\epsilon}{N_{\word{2}}(\word{v_j})+\epsilon_{2,j}}}\right],
 \]
where $\epsilon_{1,i}, \epsilon_{2,j} > 0$ are arbitrary, with $\sum\limits_{i=1}^m \epsilon_{1,i} + \sum\limits_{j=1}^{n} \epsilon_{2,j} = \epsilon$. The first term is of the correct order since $N_{\word{2}}(\word{w}) \leq N$. For the second one, we use H\"older's inequality and Lemma~\ref{lem:expYV} to obtain
{\begin{align*}
\E\left[\left|\int_0^T \left\langle  \word{u_i},\widehat{\mathbb{Y}}^U_t\right\rangle  dt\right|^{\frac{N_{\word{2}}(\word{w})+\epsilon}{N_{\word{2}}(\word{u_i})+\epsilon_{1,i}}}\right] 
    &\lesssim 
    \int_0^T \E \left|\left\langle  \word{u_i},\widehat{\mathbb{Y}}^U_t\right\rangle\right|^{\frac{N_{\word{2}}(\word{w})+\epsilon}{N_{\word{2}}(\word{u_i})+\epsilon_{1,i}}} dt
    \\ &\lesssim
    1 + \int_0^T \E \left|\left\langle  \word{u_i},\widehat{\mathbb{Y}}^U_t\right\rangle\right|^{\max\left\{\frac{N_{\word{2}}(\word{w})+\epsilon}{N_{\word{2}}(\word{u_i})+\epsilon_{1,i}}, 2\right\}} dt
    \\ & \lesssim 
    1 + \int_0^T \E \left[|U_t|^{\max\left\{N_{\word{2}}(\word{w})+\epsilon, 2(N_{\word{2}}(\word{u_i})+\epsilon_{1,i})\right\}}\right]dt
    \\ & \lesssim 
    1 + \int_0^T \E \left[|U_t|^{2(N - 1) + \epsilon}\right]dt.
\end{align*}
}
For the second term, the It\^o--Stratonovich correction is treated as in the inequality above, and for the It\^o integral, we use the Burkholder--Davis--Gundy inequality and Lemma~\ref{lem:expYV} to obtain that
{\begin{align*}
\E\left[\left|\int_0^T \left\langle  \word{v_j},\widehat{\mathbb{Y}}^U_t \right\rangle   dW_t\right|^{\frac{N_{\word{2}}(\word{w})+\epsilon}{N_{\word{2}}(\word{v_j})+\epsilon_{2,j}}}\right] &\lesssim \E\left[\left|\int_0^T \left\langle  \word{v_j},\widehat{\mathbb{Y}}^U_t \right\rangle^2   dt\right|^{\frac{N_{\word{2}}(\word{w})+\epsilon}{2(N_{\word{2}}(\word{v_j})+\epsilon_{2,j})}}\right] \\
&\lesssim 1 + \E\left[\left|\int_0^T \left\langle  \word{v_j},\widehat{\mathbb{Y}}^U_t \right\rangle^2   dt\right|^{\max\left\{\frac{N_{\word{2}}(\word{w})+\epsilon}{2(N_{\word{2}}(\word{v_j})+\epsilon_{2,j})}, 1\right\}}\right] \\
&\lesssim 1 + \E\left[ \int_0^T \left|\left\langle  \word{v_j},\widehat{\mathbb{Y}}^U_t \right\rangle\right|^{\max\left\{\frac{N_{\word{2}}(\word{w})+\epsilon}{N_{\word{2}}(\word{v_j})+\epsilon_{2,j}}, 2\right\}} dt\right]
\\
&\lesssim 1 + \E\left[ \int_0^T \left|U_t\right|^{\max\left\{{N_{\word{2}}(\word{w})+\epsilon}, 2(N_{\word{2}}(\word{v_j})+\epsilon_{2,j})\right\}} dt\right]
\\ & \lesssim 
1 + \int_0^T \E \left[|U_t|^{2(N - 1) + \epsilon}\right]dt.
\end{align*}
}
This proves \eqref{eq:term2}. 

{To conclude, combining \eqref{eq:third_term_bound}, \eqref{eq:second_term_bound}, and \eqref{eq:term2}, we observe that we have proven the existence of polynomial functions $P$ and $Q$ of orders strictly less than $2N$ and $N + 1$, respectively, such that, for any $u \in \mathcal{U}_q$,  
\begin{align*}
    {\E} \left[ \int_0^T  \left(\frac{\bar{\rho}^2 m^2 - m}{2} \left\langle \bsigma, \widehat{\mathbb{X}}^u_s \right\rangle^2 +  \rho m  \left\langle \bsigma, \widehat{\mathbb{X}}^u_s \right\rangle u_s - \frac{u_s^2}{2} \right)ds\right] \leq \E\left[\int_0^T \left(-c_1 U_t^{2N} + P(|U_t|)\right) dt +  \left(- c_4U_T^{N+1} + Q(|U_T|)\right)\right].
\end{align*}
Note that the polynomials under the expectation are bounded from above uniformly in $u$. This implies the finiteness of $\mathcal{V}(\rho,m,T)$, and by \eqref{eq:BD}, this proves the finiteness of $\E[S_T^m]$.
}

{

\subsection{Proof of Proposition \ref{prop:crit}}\label{sect:proof_moments_prop}
Without loss of generality, we assume that $\alpha > 0$ and  
$\rho = -\sqrt{1 - \frac{1}{m}} < 0.$
In this case, equation~\eqref{eq:BD} becomes  
\begin{align*}
 \ln \E\left[\frac{S_T^m}{S_0^m}\right]  
 &= \sup_{(u_t)_{t \ge 0} \in \mathcal{U}_q}  
 \E \left[ \int_0^T \left( - \frac{u_t^2}{2} - \kappa \left\langle \bsigma, \widehat{\mathbb{X}}^u_t \right\rangle u_t \right) dt \right],
\end{align*}
where $\kappa := \rho m = \sqrt{m^2 - m}$.

In the case $\bsigma = \alpha\word{222} + \beta\word{221}$, which we consider here, the objective functional can be expressed more explicitly as  
\begin{equation}
\E \left[ 
- \frac{1}{2} \int_0^T u_t^2 \, dt
- \frac{\alpha \kappa}{6} \int_0^T (B_t + U_t)^3 \, dU_t
- \frac{\beta \kappa}{2} \int_0^T \int_0^t (B_s + U_s)^2 \, ds \, dU_t 
\right].
\end{equation}

By Itô’s formula, the expectation of the middle term equals  
\[
\E\left[-\frac{\alpha\kappa}{6} \int_0^T (B_t + U_t)^3 \, dU_t\right]
= \E\left[-\frac{\alpha\kappa}{24}(B_T + U_T)^4 + \frac{\alpha\kappa}{4} \int_0^T (B_t + U_t)^2 \, dt \right],
\]
so that the objective functional becomes
\begin{equation}\label{eq:obj_fct_mom}
\E \left[
- \frac{1}{2} \int_0^T u_t^2 \, dt
- \frac{\alpha \kappa}{24}(B_T + U_T)^4
+ \frac{\alpha \kappa}{4} \int_0^T (B_t + U_t)^2 \, dt
- \frac{\beta \kappa}{2} \int_0^T \int_0^t (B_s + U_s)^2 \, ds \, dU_t
\right].
\end{equation}

\paragraph{Case $\beta = 0$.} 
To show the finiteness of moments, we take $u \in \mathcal{U}_q$ and decompose it in the following way
$$
u_t = \dfrac{1}{T}U_T + \tilde u_t, \quad U_t = \dfrac{t}{T}U_T + \tilde U_t,
$$
where $\tilde u_t := u_t - \dfrac{1}{T}U_T$ verifies $\tilde U_T = \int_0^T\tilde u_t\, dt = 0$. Using the inequality 
\begin{equation}\label{eq:ineq_quad}
    (a + b)^2 \leq (1+\epsilon) a^2 + (1+\epsilon^{-1})b^2, \quad \epsilon > 0,
\end{equation}
we can bound the third term in \eqref{eq:obj_fct_mom} as follows:
\begin{align*}
    \E\left[\frac{\alpha\kappa}{4}\int_0^T (B_t+U_t)^2 \, dt\right] 
    &= \E\left[\frac{\alpha\kappa}{4}\int_0^T \left(B_t+\dfrac{t}{T}U_T + \tilde U_t\right)^2 \, dt\right] \\
&\leq \E\left[\frac{\alpha\kappa}{4}\int_0^T \left((1 + \epsilon^{-1})\left(B_t + \dfrac{t}{T}U_T\right)^2+(1 + \epsilon)\tilde U_t^2\right) \, dt\right] \\
&\leq (1 + \epsilon^{-1})\frac{\alpha\kappa}{4}\E\left[{T^2} + \dfrac{2}{3}TU_T^2\right] + (1 + \epsilon)\frac{\alpha\kappa}{4}\E\left[\int_0^T \tilde U_t^2 \, dt\right].
\end{align*}
Since $\tilde u_t = \tilde U_t'$ and $\tilde U_0 = \tilde U_T = 0$, we can apply the Wirtinger's inequality~\citep[Chapter 7.7, 257]{hardy1952inequalities}:
\[
\int_0^T \tilde U_t^2 \, dt
\leq \frac{T^2}{\pi^2}\int_0^T \tilde u_t^2 \, dt,
\]
so that the third term is controlled by
\begin{equation}\label{eq:third_term_control}
    \E\left[\frac{\alpha\kappa}{4}\int_0^T (B_t+U_t)^2 \, dt\right] \leq C + C_{\epsilon}\E[U_T^2] + (1 + \epsilon)\frac{\alpha\kappa T^2}{4\pi^2}\E\left[\int_0^T \tilde u_t^2 \, dt\right],
\end{equation}
for some $C, C_\epsilon > 0$.

Similarly, one can control the first term in \eqref{eq:obj_fct_mom} by
\begin{equation}\label{eq:first_term_control}
    - \frac{1}{2}\E \left[\int_0^T u_t^2 \, dt  \right] \leq \left(-1 + \epsilon\right)\frac{1}{2}\E \left[\int_0^T \tilde u_t^2 \, dt  \right] + \dfrac{1}{2}\left(-1 + \epsilon^{-1}\right)\E\left[\dfrac{1}{T}U_T^2\right].
\end{equation}
Combining \eqref{eq:first_term_control} and \eqref{eq:third_term_control}, and noting that 
$\E[U_T^2]$ is dominated by the second term in \eqref{eq:obj_fct_mom}, we conclude that 
the objective functional \eqref{eq:obj_fct_mom} is bounded by  
\[
C_{0,\epsilon} + \E\!\left[\left(\!(-1+\epsilon)\frac{1}{2} + (1+\epsilon)\frac{\alpha\kappa T^2}{4\pi^2}\!\right)
\int_0^T \tilde{u}_t^2 \, dt\right],
\]
for some constant $C_{0,\epsilon} > 0$.  
The coefficient in front of $\int_0^T \tilde{u}_t^2 \, dt$ is negative if  
\(
T^2 < \left(\dfrac{1-\epsilon}{1+\epsilon}\right)\dfrac{2\pi^2}{\alpha\kappa},
\)
which holds for sufficiently small $\epsilon$ whenever  
\(
T^2 < \dfrac{2\pi^2}{\alpha\kappa}.
\)
This establishes the finiteness of the moments.

To show the explosion, consider a deterministic control 
\(
U^\lambda_t = \lambda \psi(t),
\)
where $\lambda \in \R$ and $\psi \in C^1([0,\, T])$ satisfies 
$\psi(0) = \psi(T) = 0$ and 
\(
\int_0^T \psi(t)^2 \, dt = \frac{T^2}{\pi^2}\int_0^T \psi'(t)^2 \, dt > 0.
\)
These conditions are verified, for instance, by $\psi(t) = \sin\!\left(\frac{\pi t}{T}\right)$. Hence, the objective functional \eqref{eq:obj_fct_mom} becomes
\[
\E \left[
- \frac{\lambda^2}{2}\int_0^T \psi'(t)^2 \, dt
- \frac{\alpha\kappa}{24} B_T^4
+ \frac{\alpha\kappa}{4} \int_0^T (B_t + \lambda \psi(t))^2 \, dt
\right].
\]
This expression is quadratic in $\lambda$, with leading coefficient 
\(
\left(\dfrac{\alpha\kappa T^2}{4\pi^2} - \dfrac{1}{2}\right)
\int_0^T \psi'(t)^2 \, dt.
\)
Letting $\lambda \to \infty$, we conclude that the value function diverges to infinity whenever 
\(
T^2 > \dfrac{2\pi^2}{\alpha\kappa}.
\)

\paragraph{Case $\beta \neq 0$.} 
It remains to show that the value function is infinite when $\beta \neq 0$. 
Consider the deterministic control 
\(
U^\lambda_t = \lambda \phi(t),
\)
where $\lambda \in \R$ and $\phi \in C^1([0,\, T])$ satisfies $\phi(0) = \phi(T) = 0$ and 
\(
\int_0^T \left(\int_0^t \phi(s)^2 \, ds\right) \phi'(t) \, dt \neq 0.
\)
By integration by parts, the latter condition is equivalent to
\[
\int_0^T \left(\int_0^t \phi(s)^2 \, ds\right) \phi'(t) \, dt
= \phi(T)\int_0^T \phi(s)^2 \, ds
- \int_0^T \phi(t) \, d\!\left(\int_0^t \phi(s)^2 \, ds\right)
= -\int_0^T \phi(s)^3 \, ds \neq 0.
\]

Substituting this control into \eqref{eq:obj_fct_mom}, we see that the objective functional becomes
\[
\E \left[
- \frac{\lambda^2}{2}\int_0^T \phi'(t)^2 \, dt
- \frac{\alpha\kappa}{24} B_T^4
+ \frac{\alpha\kappa}{4} \int_0^T (B_t + \lambda \phi(t))^2 \, dt
- \frac{\beta\kappa}{2} \int_0^T \left(\int_0^t (B_s + \lambda \phi(s))^2 \, ds\right)
\lambda \phi'(t) \, dt
\right].
\]
This expression is a cubic polynomial in $\lambda$, whose leading coefficient is
\(
-\dfrac{\beta\kappa}{2}
\int_0^T \left(\int_0^t \phi(s)^2 \, ds \right) \phi'(t) \, dt \neq 0,
\)
and is therefore unbounded from above as $\lambda \to \pm\infty$. 
This completes the proof.
}

\appendix
{\section{Boué-Dupuis formula and proof of (\ref{eq:BD})} \label{sec:BD}

We will use the following variant of the Boué-Dupuis formula, which is a minor extension on recent results of \cite{HW22}.

In this section, $W=(W_t)_{0\leq t \leq T}$ is a ($d$-dimensional) Brownian motion under the measure $\P$, $(\mathcal{F}_t)_{t \geq 0}$ is its augmented natural filtration, and for any $q \geq 2$, we let
\[
\mathcal{U}_q = \left\{ (u_t)_{t \geq 0} (\mathcal{F}_t)-\mbox{progressively measurable with }{\E}\left[ \int_0^T |u_t|^q dt \right] < +\infty \right\}.
\]
\begin{proposition} \label{prop:BDext}
    Let $q\geq 2$ be fixed, and $F : C([0,T],\R^d) \to \R$ be (Borel-)measurable and such that 
   \begin{equation} \label{eq:asnBDq}
        \E \left[\left| F\left(W+\int_0^{\cdot} u\right) \right|\right] < \infty, \quad  u \in \mathcal{U}_q.
   \end{equation}
Then it holds that
\[
\ln \E\left[e^{F(W)}\right] = \sup_{u \in \mathcal{U}_q} \E \left[F\left(W+\int_0^{\cdot}u\right) - \frac{1}{2} \int_0^T u_s^2 ds\right].
\]
\end{proposition}

\begin{proof}
We proceed as in \cite[Remark 1.2 (1)]{HW22} by a truncation argument. For $n \geq 0$, let $F_n = F \wedge n$, then it holds that $\E[|F_n(W)|]$ and $\E[e^{F_n(W)}]$ are finite, so that we can apply the results of \cite{HW22} to obtain
\[
\ln \E\left[e^{F_n(W)}\right] = \sup_{u \in \mathcal{U}_q} \E \left[F_n\left(W+\int_0^{\cdot}u\right) - \frac{1}{2} \int_0^T u_s^2 ds\right]
\]
(They prove that the above formula holds for the supremum taken over both $\mathcal{U}_2$ and the set of bounded controls $\mathcal{U}_b$, see Theorem 1.1 and Corollary 2.11 therein. This implies the case of arbitrary $q \geq 2$ since $\mathcal{U}_b \subset \mathcal{U}_q \subset \mathcal{U}_2$).

We then have by monotone convergence that
\begin{align*}
 \ln  \E\left[e^{F(W)}\right] &= \sup_{n} \ln \E\left[e^{F_n(W)}\right] \\
   &= \sup_{n \geq 0} \sup_{u \in \mathcal{U}_q} \E \left[F_n\left(W+\int_0^{\cdot}u\right) - \frac{1}{2} \int_0^T u_s^2 ds\right] \\
   &= \sup_{u \in \mathcal{U}_q} \sup_{n \geq 0} \E \left[F_n\left(W+\int_0^{\cdot}u\right) - \frac{1}{2} \int_0^T u_s^2 ds\right]\\
   &= \sup_{u \in \mathcal{U}_q} \E \left[F\left(W+\int_0^{\cdot}u\right) - \frac{1}{2} \int_0^T u_s^2 ds\right]
\end{align*}
where the last equality follows from \eqref{eq:asnBDq} and dominated convergence.
\end{proof}

\begin{proof}[Proof of \eqref{eq:BD}]
    We let 
    \[
    F(W)=\rho m \int_0^T \bracketsig[s]{\bsigma} \,d W_s + \frac{\bar{\rho}^2 m^2 - m}{2} \int_0^T \bracketsig[s]{\bsigma}^2 \, ds
    \]
    and check that it satisfies \eqref{eq:asnBDq} for $q \geq 2N$.
    For $u \in \mathcal{U}_q$, letting $\widehat{\mathbb{X}}^u_t$ denote the signature of $\widehat{X}^u_t = (t, X^u_t)$ with $X^u_t$ defined by
\[
X^u_t = W_t + \int_0^{t} u_s ds,
\]
one has that
\[
F\left(W+ \int_0^{\cdot} u\right) =\rho m \int_0^T \left\langle \bsigma, \widehat{\mathbb{X}}^u_s \right\rangle\,(d W_s+u_s ds) + \frac{\bar{\rho}^2 m^2 - m}{2} \int_0^T \left\langle \bsigma, \widehat{\mathbb{X}}^u_s \right\rangle^2 \, ds
\]
As in Section 10.2, there exists $\tilde{\bsigma} \in \tTA[3]{N}$, such that $\left\langle \bsigma, \widehat{\mathbb{X}}^u_t \right\rangle = \left\langle \tilde\bsigma,\widehat{\mathbb{Y}}^U_t \right\rangle$ where $\widehat{\mathbb{Y}}^U_t$ is the signature of $Y^U_t=(t, \int_0^{t} u_s ds, W_t)$. By Lemma \ref{lem:expYV} and Hölder's inequality, this implies that
\begin{equation} \label{eq:boundBDq}
  \E \left[\int_0^T \left\langle \bsigma, \widehat{\mathbb{X}}^u_s \right\rangle^2 \, ds\right] \leq C \left( 1 + \E\left[\int_0^T |u_s|^{2N} ds\right] \right) < \infty  
\end{equation}
since $q \geq 2N$. Since it holds that
\[\E \left[ \left| F\left(W+ \int_0^{\cdot} u\right) \right| \right] \leq C \left( 1 + \E\left[ \int_0^T u_s^2 ds\right] + \E\left[ \int_0^T \left\langle \bsigma, \widehat{\mathbb{X}}^u_s \right\rangle^2  ds\right]\right)
\]
(using It\^o isometry for the term in $dW$ and Young's inequality for the term in $u_s ds$), we see that \eqref{eq:asnBDq} holds.

We then use Proposition \ref{prop:BDext} to obtain that
\begin{align*}
   \ln \E\left[e^{F(W)}\right] &=  \sup_{u \in \mathcal{U}_q} {\E} \left[- \int_0^T\frac{u_s^2}{2} ds + \int_0^T \rho m  \left\langle \bsigma, \widehat{\mathbb{X}}^u_s \right\rangle (dW_s + u_sds) +   \frac{\bar{\rho}^2 m^2 - m}{2}\int_0^T \left\langle \bsigma, \widehat{\mathbb{X}}^u_s \right\rangle^2ds\right] \\
&=  \sup_{u \in \mathcal{U}_q} {\E} \left[- \int_0^T\frac{u_s^2}{2} ds + \int_0^T \rho m  \left\langle \bsigma, \widehat{\mathbb{X}}^u_s \right\rangle  u_s ds +   \frac{\bar{\rho}^2 m^2 - m}{2}\int_0^T \left\langle \bsigma, \widehat{\mathbb{X}}^u_s \right\rangle^2ds\right],
\end{align*}
where we have used that the stochastic integral has zero expectation by \eqref{eq:boundBDq}.
\end{proof}
}

\bibliographystyle{plainnat}
\bibliography{main}

\end{document}